\documentclass[12pt]{amsart}
\usepackage{amsmath}
\usepackage{latexsym}
\usepackage{amsfonts}
\usepackage{amssymb}
\usepackage{color}
\usepackage{bbm,dsfont}
\usepackage{graphicx}

\newtheorem{proposition}{Proposition}
\newtheorem{theorem}{Theorem}
\newtheorem{lemma}{Lemma}
\newtheorem{corollary}{Corollary}

\theoremstyle{definition}

\newtheorem{remark}{Remark}
\newtheorem{example}{Example}

\newcommand{\R}{\mathbb{R}} 
\newcommand{\C}{\mathbb{C}} 

\newcommand{\nat}{\mathbb N} 
\newcommand{\Z}{\mathbb Z} 
\newcommand{\half}{\tfrac{1}{2}} 
\newcommand{\mo}[1]{\left| #1 \right|} 

\newcommand{\frecc}{\to}

\newcommand{\hi}{\mathcal{H}} 
\newcommand{\hh}{\mathcal{H}} 
\newcommand{\lh}{\mathcal{L(H)}} 
\newcommand{\sh}{\mathcal{S(H)}} 
\newcommand{\ip}[2]{\left\langle\,#1\,|\,#2\,\right\rangle} 
\newcommand{\kb}[2]{|#1\rangle\langle#2|} 
\newcommand{\no}[1]{\left\|#1\right\|} 
\newcommand{\tr}[1]{{\rm tr}\left[#1\right]} 
\newcommand{\id}{\mathbbm{1}} 
\newcommand{\fii}{\varphi}
\newcommand{\lam}{\lambda}
\newcommand{\Lam}{\Lambda}

\newcommand{\lft}{\left(}
\newcommand{\rgt}{\right)}

\newcommand{\ff}{\mathcal{F}}

\newcommand{\Ao}{\mathsf{A}}
\newcommand{\Bo}{\mathsf{B}}
\newcommand{\Co}{\mathsf{C}}
\newcommand{\Go}{\mathsf{G}}

\newcommand{\Ec}{\mathcal{E}} 

\newcommand{\Ii}{\mathcal{I}}
\newcommand{\Ji}{\mathcal{J}}

\begin{document}\setlength{\arraycolsep}{2pt}

\title[Informationally complete joint measurements]{Informationally complete joint measurements on finite quantum systems}

\author[Carmeli]{Claudio Carmeli}
\thanks{Claudio Carmeli, Dipartimento di Fisica, Universit\`a di Genova, Via Dodecaneso 33, I-16146 Genova, Italy\\ 
email: claudio.carmeli@gmail.com}

\author[Heinosaari]{Teiko Heinosaari}
\thanks{Teiko Heinosaari, Turku Centre for Quantum Physics, Department of Physics and Astronomy, University of Turku\\
email: teiko.heinosaari@utu.fi}

\author[Toigo]{Alessandro Toigo}
\thanks{Alessandro Toigo, Dipartimento di Matematica, Politecnico di Milano, Piazza Leonardo da Vinci 32, I-20133 Milano, Italy, and I.N.F.N., Sezione di Milano, Via Celoria 16, I-20133 Milano, Italy\\
email: alessandro.toigo@polimi.it}


\begin{abstract}
We show that there are informationally complete joint measurements of two conjugated observables on a finite quantum system, meaning that they enable to identify all quantum states from their measurement outcome statistics.
We further demonstrate that it is possible to implement a joint observable as a sequential measurement.
If we require minimal noise in the joint measurement, then the joint observable is unique.
If the dimension $d$ is odd, then this observable is informationally complete.
But if $d$ is even, then the joint observable is not informationally complete and one has to allow more noise in order to obtain informational completeness.
\end{abstract}

\maketitle

\section{Introduction}\label{sec:intro}

The general aim in quantum tomography is to identify quantum states from measurement outcome statistics.
A collection of observables with this property is called \emph{informationally complete} \cite{Prugovecki77}.
Even a single observable can be informationally complete, but then it must be a noncommutative positive operator valued measure (POVM) \cite{BuLa89,Busch91}.
We will study a class of informationally complete POVMs in dimension $d$ with the minimal number of $d^2$ outcomes and we will explain how they can be implemented as sequential measurements of two $d$-outcome measurements.
One can interpret the generated joint observable as a phase space measurement in the discrete phase space $\Z_d\times\Z_d$ \cite{OpBuBaDr95,Vourdas97}.

There are some particularly interesting approaches to finite dimensional quantum tomography, and one of them is based on complete collections of mutually unbiased bases (MUBs) \cite{Ivanovic81, Wootters86,WoFi89}.
In a $d$-dimensional Hilbert space one needs $d+1$ MUBs in order to be able to identify all quantum states, but it is not known if a complete set of MUBs exists in all dimensions.
In fact, there is evidence that for $d=6$ there is no complete set of MUBs \cite{BrWe08,RaLuEn11}. 

In our scheme we start from two mutually unbiased bases connected by the finite Fourier transform.
They define a pair of complementary observables, which cannot be measured jointly. 
However, it is possible to realize their joint measurement if some additional noise is allowed.
We show that their joint measurement can be chosen to be informationally complete, and that this can be realized as a sequential measurement where we first perform a `weak measurement' in one basis and then another successive measurement in the other basis.
Compared to the fact that one would need $d+1$ complementary observables in order to reach informational completeness in separate measurements, it is remarkable that in the sequential scheme only two observables suffice.  

The price to have a joint measurement is that the marginal observables are not the original complementary observables but their unsharp versions.
We will analyze the required additional noise and characterize the optimal joint observable from this point of view.
The qubit case has been first studied in \cite{Busch86}, and our work generalizes those results to arbitrary finite dimension.

The covariant phase space observables, i.e., POVMs covariant under the finite Weyl-Heisenberg group, play a special role in our investigation. We prove that if a pair of conjugate observables have a joint measurement, then they also have a joint measurement which is a covariant phase space observable.
Since every covariant phase space observable arises from a sequential measurement of two conjugate observables \cite{CaHeTo11}, the covariant phase space observables are an outstanding choice for finite dimensional quantum tomography.

The Weyl-Heisenberg group has also a pivotal role in the investigations of symmetric informationally complete (SIC) observables \cite{ReBlScCa04,Appleby05JMP}.
It is generally believed that a Weyl-Heisenberg covariant SIC observable exists in every finite dimension and their existence is numerically tested in all dimensions up to 67 \cite{ScGr10}.
Our results show that any such observable has a neat sequential realization scheme.

There is an interesting difference between the even and odd dimensional Hilbert spaces.
If we require minimal noise in both marginal observables, then their joint observable is unique.
If $d$ is odd, then this observable is informationally complete.
But if $d$ is even, then the joint observable is not informationally complete.
This result gives an additional aspect to the common observation that quantum tomography is different in even and odd dimensions \cite{ShVo11}.

\section{Preliminaries}\label{sec:preli}

In this section we fix some notations and introduce the basic concepts.

\subsection*{States and Observables}

Let $\hi$ be a finite dimensional Hilbert space, with $\dim\hi = d \geq 2$.
We denote by $\lh$ the vector space of all linear operators on $\hi$.
A positive operator $\varrho\in\lh$ having trace one is a \emph{state}, and we denote by $\sh$ the set of all states.

Observables are generally described by \emph{positive operator valued measures} (POVMs) \cite{OQP97,FQMEA02}.
In this work we only consider observables with finite number of outcomes.
Therefore, an observable can be defined as a function $\Ao:x\mapsto\Ao(x)$, where each $\Ao(x)$ is a positive operator and $\sum_x \Ao(x) = \id$. 
Here the sum runs over all $x\in\Omega_\Ao$, where the set $\Omega_\Ao$ is the collection of all possible measurement outcomes.

If a system is prepared in a state $\varrho$, then a measurement of an observable $\Ao$ will lead to an outcome $x$ with the probability $\tr{\varrho\Ao(x)}$.

\subsection*{Informational completeness}

An observable $\Ao$ is \emph{informationally complete} if its measurement outcome probability distribution is sufficient to identify a unique state \cite{Prugovecki77}.
In other words, two different states must give rise to different probability distributions: for all pairs of states $\varrho_1,\varrho_2$,
\begin{equation*}
\tr{\varrho_1 \Ao(j)} = \tr{\varrho_2 \Ao(j)} \quad \forall j\in\Omega_\Ao \quad \Rightarrow \quad\varrho_1=\varrho_2 \, .
\end{equation*}

The informational completeness of an observable $\Ao$ is equivalent to the property that the linear span of the set $\{\Ao(j):j\in\Omega_\Ao\}$ is $\lh$ \cite{Busch91,SiSt92}.

\subsection*{Joint measurability}

Given two observables $\Ao$ and $\Bo$, we say that they are \emph{jointly measurable} if there exists a third observable $\Co$ with $\Omega_\Co=\Omega_\Ao \times\Omega_\Bo$ and satisfying 
\begin{equation}\label{eq:joint-definition}
\sum_{x\in\Omega_\Ao}\Co(x,y)=\Bo(y) \quad \forall y \, , \qquad \sum_{y\in\Omega_\Bo} \Co(x,y)=\Ao(x) \quad \forall x \, .
\end{equation}
In other words, $\Ao$ and $\Bo$ correspond to the `marginals' of $\Co$.
Any observable satisfying \eqref{eq:joint-definition} is called a \emph{joint observable} of $\Ao$ and $\Bo$ \cite{LaPu97}.

We recall that joint measurability is equivalent to the following \cite{AlCaHeTo09}: there exists an observable $\Go$ and stochastic matrices $[M_{xz}]$, $[M'_{yz}]$ such that
\begin{equation}
\sum_z M_{xz} \Go(z)=\Ao(x) \quad \forall x\, , \quad \sum_z M'_{yz} \Go(z)=\Bo(y) \quad \forall y \, .
\end{equation}
Hence, two observables are jointly measurable iff they can be `post-processed'  from a single observable.

We will use several times the following simple fact: if $\Co$ and $\Co'$ are joint observables of $\Ao$ and $\Bo$, then also all their convex combinations $t\Co+(1-t)\Co'$, $0<t<1$, are joint observables of $\Ao$ and $\Bo$.
It follows that two jointly measurable observables have either a unique joint observable or infinitely many of them.

Another useful fact is related to unitary transformations.
Let $U$ be a unitary operator on $\hi$.
Two observables $\Ao$ and $\Bo$ are jointly measurable if and only if the observables $U\Ao U^\ast$ and $U \Bo U^\ast$ are jointly measurable.
Indeed, it is easy to see that $\Co$ is a joint observable of $\Ao$ and $\Bo$ if and only if $U\Co U^\ast$ is a joint observable of $U\Ao U^\ast$ and $U\Bo U^\ast$.

\subsection*{Instruments}

An observable describes the statistics of the outcomes of a measurement but leaves open how the measurement disturbs the input state.
In order to discuss this we need the concept of an \emph{instrument} \cite{QTOS76}.
An instrument with finitely many outcomes is a mapping $\Ii:x\mapsto\Ii_x$ such that each $\Ii_x$ is 
a completely positive linear map on $\lh$ and $\sum_x \tr{\Ii_x(\varrho)}=1$ for all states $\varrho$.

The \emph{adjoint map} $\Ii_x^\ast$ of $\Ii_x$ is defined via the usual trace duality 
\begin{equation*}
\tr{S \Ii_x(T)}=\tr{\Ii_x^\ast(S)T} \qquad \forall S,T\in\lh \, .
\end{equation*}
In other words, $\Ii_x^\ast$ and $\Ii_x$ correspond to the Heisenberg and Schr\"odinger pictures, respectively.

Suppose that $\Ao$ is an observable.
Then we say that an instrument $\Ii$ is \emph{$\Ao$-compatible} if $\Ii_x^\ast(\id)=\Ao(x)$ for every $x$.
Every $\Ao$-compatible instrument describes some particular kind of measurement of $\Ao$ \cite{Ozawa84}.

An example of an $\Ao$-compatible instrument is the \emph{L\"uders instrument} $\Ii^L$, defined by
\begin{equation*}
\Ii^L_x(\varrho)=\sqrt{\Ao(x)} \varrho \sqrt{\Ao(x)} \, .
\end{equation*}
Any other $\Ao$-compatible instrument $\Ii$ is of the form
\begin{equation*}
\Ii_x(\varrho) = \Ec_x \left( \Ii^L_x(\varrho) \right)
\end{equation*}
for some collection $\{\Ec_x\}$ of completely positive trace preserving maps on $\lh$ \cite{QI06}.

\subsection*{Sequential measurements}

By a \emph{sequential measurement} we mean a setting where two measurements are combined
into a third measurement by performing them one after the other \cite{DaLe70}.
Generally, the order in which the measurements are performed is crucial \cite{BuCaLa90}.

Suppose we have two $N$-outcome observables $\Ao,\Bo$ and we measure them subsequently; first $\Ao$ and then $\Bo$.
As a result, we have in total $N^2$ possible measurement outcomes. 
Generally, we do not obtain a joint measurement of $\Ao$ and $\Bo$ since the first measurement distrubs the input state.
In fact, the overall measurement depends on the way we measure $\Ao$.
If the first measurement is described by an $\Ao$-compatible instrument $\Ii$, then 
the overall observable $\Co$ is given by 
\begin{equation*}
\tr{\varrho \Co(j,k)} = \tr{\Bo(k)\Ii_j(\varrho)} 
\end{equation*}
for all input states $\varrho$, or equivalently, 
\begin{equation*}
\Co(j,k) = \Ii_j^\ast(\Bo(k)) \, .
\end{equation*}
Let us notice that first marginal of $\Co$ is always $\Ao$, while the second marginal is a perturbed version of $\Bo$ and depends on the instrument $\Ii$.

\section{Example: sequential measurements of $\sigma_x$ and $\sigma_y$}\label{sec:qubit}

We start with a preliminary example, which is mainly a collection of well known facts.
It hints the forthcoming developments and clarifies the aims of the later sections.
We refer to \cite{BuHe08} for more details and further references.

Fix $\hi=\C^2$, and let $\Ao$ and $\Bo$ be the two observables corresponding to the measurements of spin-$\half$ components in the directions $x$ and $y$, respectively. 
Thus, 
\begin{equation*}
\Ao(\pm 1) = \half (\id\pm\sigma_x) \, , \quad \Bo(\pm 1) = \half (\id\pm\sigma_y) \, ,
\end{equation*}
where $\sigma_x,\sigma_y$ are the Pauli spin matrices.

Since $\Ao$ and $\Bo$ consist of projections and they do not mutually commute, it is not possible to measure them jointly.
Moreover, if we measure them separately on two similarly prepared ensembles, we still cannot infer the unknown state.

An alternative way is to perform a sequential measurement.
The first measurement has to be a weak measurement, meaning that we do not measure $\Ao$ but its unsharp version. 
We define an unsharp version $\Ao_\lambda$ of $\Ao$ by 
\begin{equation*}
\Ao_\lambda(j) := \lambda \Ao(j) + (1-\lambda) \frac{1}{2} \id \, , \quad j=\pm 1 \, .
\end{equation*}
Here $\lambda\in[0,1]$ is a parameter quantifying the noise or imprecision.
We can write $\Ao_\lambda$ in the form
\begin{equation*}
\Ao_\lambda(\pm 1)= \half \bigl( \id \pm \lambda \sigma_x \bigr) \, .
\end{equation*}
In a similar way we define an unsharp version $\Bo_\gamma$ of $\Bo$ by 
\begin{equation*}
\Bo_\gamma(\pm 1)= \half \bigl( \id \pm \gamma \sigma_y \bigr) \, .
\end{equation*}

We want to study the disturbance of the first measurement on the system, and for this reason we define an instrument related to $\Ao_\lambda$. 
A class of $\Ao_\lambda$-compatible instruments can be defined by
\begin{equation*}
\Ji_{\pm 1}(\varrho) = L_{\pm 1}\sqrt{\Ao_\lambda(\pm 1)}\varrho\sqrt{\Ao_\lambda(\pm 1)}L_{\pm 1}^\ast \, ,
\end{equation*}
where $L_1,L_{-1}$ are arbitrary unitary operators.
If the subsequent measurement is a $\Bo$-measurement, then the overall statistics of the sequential measurement is given by the observable
\begin{equation*}
\Co(j,k) = \Ji_j^\ast(\Bo(k)) = \sqrt{\Ao_\lambda(j)}L_j^\ast \Bo(k) L_j \sqrt{\Ao_\lambda(j)}   \, , \quad j,k=\pm 1 \, .
\end{equation*}
The properties of $\Co$ obviously depend on $L_1$ and $L_{-1}$.
In the following we consider two different choices of $L_{\pm 1}$.

\subsection*{Optimal joint measurement}

If we choose $L_{\pm 1}=\id$, then we obtain
\begin{equation*}
\Co(j,k) =\frac{1}{4} \bigl( \id + j\ \lambda \sigma_x + k\ \sqrt{1-\lambda^2} \sigma_y \bigr) \, , \quad j,k=\pm 1 \, .
\end{equation*}
In particular, the marginals are
\begin{eqnarray*}
\Co(j,+1) + \Co(j,-1) &=& \Ao_\lambda(j) \, ,  \\
\Co(+1,k) + \Co(-1,k) &=& \frac{1}{2} \bigl( \id + k\ \sqrt{1-\lambda^2}\ \sigma_y \bigr) = \Bo_{\sqrt{1-\lambda^2}}(k) \, .\\
\end{eqnarray*}
The joint observable $\Co$ is an optimal approximate joint measurement of $\sigma_x$ and $\sigma_y$.
This means that the unsharp parameters $\lambda$ and $\gamma=\sqrt{1-\lambda^2}$ saturate the inequality
\begin{equation}\label{eq:ur-qubit-2}
\lambda^2 + \gamma^2   \leq 1 \, .
\end{equation}
Indeed, it is known that this inequality is a necessary and sufficient criterion for two observables $\Ao_\lambda$ and $\Bo_\gamma$ to be jointly measurable \cite{Busch86}.
Let us also notice that the joint observable of $\Ao_\lambda$ and $\Bo_\gamma$ is unique if $\lambda^2 + \gamma^2   = 1$ \cite{Busch86}.

\subsection*{Informationally complete joint measurement}

Another interesting option is to choose
\begin{equation*}
L_{\pm 1} = \cos\frac{\theta}{2}\id \mp i \sin\frac{\theta}{2}\sigma_x 
\end{equation*}
for some fixed angle $0<\theta < \pi/2$.
In this case we obtain
\begin{align*}
\Co(j,k)  &= \frac{1}{4} \bigl( \id +j \ \lambda\sigma_x + k \ \cos\theta\sqrt{1-\lambda^2} \sigma_y + jk   \ \sin\theta\sqrt{1-\lambda^2} \sigma_z\bigr)
\end{align*}
and the marginals are
\begin{eqnarray*}
\Co(j,+1) + \Co(j,-1) &=& \Ao_\lambda(j) \, , \\
\Co(+1,k) + \Co(-1,k) &=& \frac{1}{2} \bigl( \id + k \ \cos\theta\sqrt{1-\lambda^2} \sigma_y \bigr) = \Bo_{\cos\theta\sqrt{1-\lambda^2}}(k) \, .
\end{eqnarray*}

It is easy to see that the linear span of the four operators $\Co(j,k)$, $j,k=\pm1$, is the set of all $2\times 2$ - complex matrices.
It follows that the joint observable $\Co$ is informationally complete.

The unsharpness parameters $\lambda$ and $\gamma=\cos\theta\sqrt{1-\lambda^2}$ do not saturate the inequality \eqref{eq:ur-qubit-2}.
Altering the parameter $\theta$ we can make the sum $\lambda^2+\gamma^2$ as close to $1$ as we want, hence we conclude that
$\Ao_\lambda$ and $\Bo_\gamma$ admit an informationally complete joint observable if and only if
\begin{equation}
\lambda^2 + \gamma^2   < 1 \, .
\end{equation}
 
Finally, we remark that with the choices $\lambda=1/\sqrt{3}$ and $\theta=\pi/4$ the joint observable $\Co$ is a symmetric informationally complete (SIC) observable.

\section{Conjugate observables}\label{sec:conjugate}

\subsection{Mutually unbiased bases and complementary observables}\label{sec:mub}

We start by recalling the usual definition of complementary observables in a finite $d$-dimensional Hilbert space and some related basic facts \cite{Schwinger60}, \cite{Kraus87}.
We denote $\Z_d\equiv \{0,\ldots,d-1\}$.
Let $\{\varphi_j\}_{j\in\Z_d}$ and $\{\psi_k\}_{k\in\Z_d}$ be \emph{mutually unbiased bases} (MUBs), i.e., they are orthonormal bases in $\hi$ and
\begin{equation}\label{eq:mub}
\mo{\ip{\varphi_j}{\psi_k}}^2 = 1/d \qquad \forall j,k \in\Z_d \, . 
\end{equation}
We define two $d$-outcome observables $\Ao$ and $\Bo$ corresponding to $\{\varphi_j\}_{j\in\Z_d}$ and $\{\psi_k\}_{k\in\Z_d}$, respectively.
Hence, 
\begin{equation}\label{eq:AB}
\Ao(j)=\kb{\varphi_j}{\varphi_j} \, , \qquad \Bo(k)=\kb{\psi_k}{\psi_k} \, .
\end{equation}
Obviously, two orthonormal bases $\{\varphi_j\}_{j\in\Z_d}$ and $\{\varphi'_j\}_{j\in\Z_d}$ define the same observable $\Ao$ iff $\varphi'_j=\alpha_j \varphi_j$ for some complex numbers $\alpha_j$ of modulus one.

The mutual unbiasedness condition \eqref{eq:mub} can be rephrased by saying that $\Ao$ and $\Bo$ are \emph{complementary} observables, meaning that in any state $\varrho$ where the outcome of $\Ao$ is predictable, the $\Bo$-distribution is uniform (and vice versa).
This entails that the following implications are valid for any state $\varrho$ and all outcomes $j,k\in\Z_d$,
\begin{align*}
\tr{\varrho\Ao(j)} &= 1 \quad\Rightarrow\quad \tr{\varrho\Bo(k)} = 1/d \\
\tr{\varrho\Bo(k)} &= 1 \quad\Rightarrow\quad \tr{\varrho\Ao(j)} = 1/d \, .
\end{align*}
Since $\tr{\varrho\Ao(j)} = 1$ iff $\varrho=\kb{\varphi_j}{\varphi_j}$, it is easy to see that the complementarity of $\Ao$ and $\Bo$ is indeed equivalent to the mutual unbiasedness of the bases $\{\varphi_j\}_{j\in\Z_d}$ and $\{\psi_k\}_{k\in\Z_d}$.

There is a canonical way to produce two mutually unbiased bases.
In the following, suppose an orthonormal basis $\{\fii_k\}_{k\in\Z_d}$ of $\hh$ is fixed. 
Denoting $\omega\equiv e^{2\pi i /d}$, we define the following unitary representations $U$ and $V$ of the cyclic group $\Z_d$ in $\hh$:
\begin{eqnarray*}
U_x \fii_k & := & \fii_{k+x} \\
V_y \fii_k & := & \omega^{yk} \fii_k
\end{eqnarray*}
for all $x,y,k \in\Z_d$. 
In the above formulas and in the rest of the paper, addition and multiplication of elements in $\Z_d$ are understood modulo $d$.
(For instance, we will often use $-j=d-j$).
It is easy to verify that
\begin{equation}\label{eq:VU}
V_y U_x = \omega^{xy} \ U_x V_y \qquad \forall x,y\in\Z_d \, .
\end{equation}

The Fourier transform (with respect to the basis $\{\fii_k\}_{k\in\Z_d}$) is the unitary operator $\ff : \hh \frecc \hh$ defined by
\begin{equation}\label{eq:defF}
\ff \fii_k := \frac{1}{\sqrt{d}} \sum_{h\in\Z_d} \omega^{-hk} \fii_h \, .
\end{equation}
The adjoint operator $\ff^\ast$ of $\ff$ is given by
\begin{equation}\label{eq:defF-adjoint}
\ff^\ast \fii_k = \frac{1}{\sqrt{d}} \sum_{h\in\Z_d} \omega^{hk} \fii_h = \ff\fii_{-k} \, ,
\end{equation}
and we have $\ff^2 \fii_k = \ff^{\ast \, 2} \fii_k = \fii_{-k}$. 
We denote 
\begin{equation*}
\psi_k \equiv \ff^\ast \fii_k = \ff \fii_{-k} \, , 
\end{equation*}
and it is immediate to check that $\{\fii_j\}_{j\in\Z_d}$ and $\{\psi_k\}_{k\in\Z_d}$ are MUBs, with $\ip{\fii_j}{\psi_k} = (1/\sqrt{d}) \, \omega^{jk}$.

The Fourier transform has the intertwining properties
$$
\ff U_x = V_x^\ast \ff \, , \qquad \ff V_y = U_y \ff \, ,
$$
from which it follows that
\begin{eqnarray*}
U_x \psi_k & = & \omega^{-xk} \psi_k \\
V_y \psi_k & = & \psi_{k+y} \, .
\end{eqnarray*}

The observables $\Ao$ and $\Bo$ related to $\{\fii_j\}_{j\in\Z_d}$ and $\{\psi_k\}_{k\in\Z_d}$, respectively, satisfy the following conditions for all $j,k,x,y\in\Z_d$:
\begin{equation}\label{eq:cov-A}
U_x \Ao(j) U_x^\ast = \Ao(j+x) \, , \quad V_y \Ao(j) V_y^\ast = \Ao(j) 
\end{equation}
and
\begin{equation}\label{eq:cov-B}
U_x \Bo(k) U_x^\ast = \Bo(k) \, , \quad V_y \Bo(k) V_y^\ast = \Bo(k+y) \, . 
\end{equation}
In other words, $\Ao$ is $U$-covariant and $V$-invariant, while $\Bo$ is $U$-invariant and $V$-covariant. 
We also note that $\Ao$ and $\Bo$ are conjugated by $\ff$, i.e., 
\begin{equation}
\Bo(k) = \ff^\ast \Ao(k) \ff
\end{equation}
 for all $k\in\Z_d$. 
It is customary to say that $\Ao$ and $\Bo$ are {\em canonically conjugated} observables.

The conditions \eqref{eq:cov-A} -- \eqref{eq:cov-B} are analogous to the symmetry properties of the usual position and momentum observables on the real line $\R$ (see e.g. \cite{CaHeTo04}).
In some situations the covariance properties may have some physical meaning or motivation.
However, for our purposes they are just useful features that can be utilized later in our calculations.

\subsection{Unsharp observables}\label{sec:unsharp}

Measurements of two complementary observables are incompatible and therefore have to be performed separately. 
This means that their measurements require different settings.
However, it is possible to perform a simultaneous measurement of two complementary observables if we allow some additional imprecision or noise.
In other words, we can measure jointly \emph{unsharp versions} of $\Ao$ and $\Bo$.

We define an unsharp version $\Ao_\lambda$ of $\Ao$ by
\begin{equation*}
\Ao_\lambda(j) := \lambda \Ao(j) + (1-\lambda) \frac{1}{d} \id \, .
\end{equation*}
Here $\lambda\in[0,1]$ is a parameter quantifying the noise.
This type of noise is equivalent to the situation where an input state $\varrho$ is first depolarized into a state $\lambda \varrho + (1-\lambda)/d \id$ and then a measurement of $\Ao$ is performed.

More generally, if $\Lambda$ is a probability distribution on $\Z_d$, then we define an unsharp version $\Ao_\Lambda$ of $\Ao$ by 
\begin{equation}\label{eq:defALam}
\Ao_\Lambda(j) := \sum_{i\in\Z_d} \Lambda(j-i) \Ao(i) \, . 
\end{equation}
The special case $\Ao_\Lambda=\Ao_\lambda$ corresponds to the probability distribution $\Lambda$ defined as
\begin{equation*}
\Lambda(0) =  \lambda + (1-\lambda)/d \, , \quad  \Lambda(j) = (1-\lambda)/d \quad \textrm{if}\ j\neq 0 \, .
\end{equation*}
We can also write this probability distribution in the form
\begin{equation*}
\Lambda(j) = \lambda \delta (j) + (1-\lambda) \mu (j) 
\end{equation*}
where $\delta$ is the point distribution at $0$ and $\mu$ is the uniform distribution on $\Z_d$, i.e.,
\begin{equation}\label{eq:point-and-uniform}
 \delta (j) = \left\{\begin{array}{cc} 1 & \quad \mbox{if } j=0 \\ 0 & \quad \mbox{if } j\neq 0 \end{array} \right. \, , \qquad \mu (j) = \frac{1}{d} \quad \forall j \, .
\end{equation}

In a similar way a probability distribution $\Gamma$ on $\Z_d$ defines an unsharp version $\Bo_\Gamma$ of $\Bo$ by 
\begin{equation}\label{eq:defBGamma}
\Bo_\Gamma(k) := \sum_{i\in\Z_d} \Gamma(k-i) \Bo(i) \, . 
\end{equation}
A special class is, again, characterized by noise parameters $\gamma \in [0,1]$ and we denote
\begin{equation*}
\Bo_\gamma(k) := \gamma \Bo(k) + (1-\gamma) \frac{1}{d} \id \, .
\end{equation*}

Naturally, there are also other type of approximations of $\Ao$ and $\Bo$ than the previously defined $\Ao_\Lambda$ and $\Bo_\Gamma$.
The usefulness of $\Ao_\Lambda$ and $\Bo_\Gamma$ is that they satisfy the same covariance and invariance relations than $\Ao$ and $\Bo$, respectively.
Namely, the observables $\Ao_\Lambda$ and $\Bo_\Gamma$ satisfy the following conditions:
\begin{equation}\label{eq:cov-A-unsharp}
U_x \Ao_\Lambda(j) U_x^\ast = \Ao_\Lambda(j+x) \, , \quad V_y \Ao_\Lambda(j) V_y^\ast = \Ao_\Lambda(j) 
\end{equation}
and
\begin{equation}\label{eq:cov-B-unsharp}
U_x \Bo_\Gamma(k) U_x^\ast = \Bo_\Gamma(k) \, , \quad V_y \Bo_\Gamma(k) V_y^\ast = \Bo_\Gamma(k+y) \, . 
\end{equation}
Thus, $\Ao_\Lambda$ and $\Bo_\Gamma$ are \emph{conjugated observables} although they need not be complementary anymore \cite{CaHeTo11}.
As we will see, two observables $\Ao_\Lambda$ and $\Bo_\Gamma$ can have a joint observable even if they do not commute.

\begin{remark}
Suppose that $\widetilde{\Ao}$ is a $d$-outcome observable satisfying 
\begin{equation}\label{eq:cov-A-general}
U_x \widetilde{\Ao}(j) U_x^\ast = \widetilde{\Ao}(j+x) \, , \quad V_y \widetilde{\Ao}(j) V_y^\ast = \widetilde{\Ao}(j) 
\end{equation}
for all $j,x,y\in\Z_d$.
Then $\widetilde{\Ao}=\Ao_\Lambda$ for some probability distribution $\Lambda$.
Namely, it follows from the second condition in \eqref{eq:cov-A-general} that $\widetilde{\Ao}$ commutes with $\Ao$ (since $\Ao(j) = (1/d) \sum_y \omega^{-jy} V_y$) and hence $\widetilde{\Ao}(j)=\sum_k p_{j,k} \Ao(k)$ for some real numbers $0\leq p_{j,k} \leq 1$.
The first condition in \eqref{eq:cov-A-general} then implies that $p_{j,k}=p_{0,k-j}$.
\end{remark}

\section{Joint measurements}\label{sec:joint}

\subsection{Covariant observables}

We recall that two observables $\Ao_\Lambda$ and $\Bo_\Gamma$ are jointly measurable if they have a joint observable, i.e., an observable $\Co$ on $\Z_d\times\Z_d$ such that
\begin{equation}\label{eq:joint-povm}
\sum_{k\in\Z_d} \Co(j,k) = \Ao_\Lambda(j) \, , \qquad \sum_{j\in\Z_d} \Co(j,k) = \Bo_\Gamma(k)
\end{equation}
for all $j,k\in\Z_d$.
A special class of joint observables turns out to be crucial for our developments. 
We say that an observable $\Co$ on $\Z_d\times\Z_d$ is a \emph{covariant phase space observable} if
\begin{equation}
U_xV_y\Co(j,k)V_y^\ast U_x^\ast = \Co(j+x,k+y) 
\end{equation}
for all $j,k,x,y\in\Z_d$.
The covariant phase space observables have a simple form \cite{QTOS76}.
Namely, if $\Co$ is a covariant phase space observable, then there is unique operator $T\in\sh$ such that $\Co=\Co_T$, where we have denoted
\begin{equation*}
\Co_T(j,k):=\frac{1}{d}\ U_jV_k T V_k^\ast U_j^\ast \, , \qquad j,k\in\Z_d \, .
\end{equation*}
Also, each $T\in\sh$ defines a covariant phase space observable by this formula.
The correspondence $T\leftrightarrow \Co_T$ is therefore one-to-one and the elements in $\sh$ parametrize the covariant phase space observables.

The marginals of a covariant phase space observable $\Co_T$ are conjugated observables on $\Z_d$.
Indeed, a direct calculation shows that $\Co_T$ has marginals $\Ao_\Lambda$ and $\Bo_\Gamma$, with
\begin{equation}\label{eq:structure1}
\Lambda (j) = \tr{\Ao (-j)T} \, , \qquad
\Gamma (k) = \tr{\Bo (-k)T} \, .
\end{equation}
(This calculation can be found in \cite{CaHeTo11}).

The essential role of covariant phase space observables in our discussion becomes clear in the following observation.

\begin{proposition}\label{prop:if-jm-then-cov}
If $\Ao_\Lambda$ and $\Bo_\Gamma$ are jointly measurable, then they have a joint observable which is a covariant phase space observable.
\end{proposition}

\begin{proof}
Suppose that $\Co$ is a joint observable of $\Ao_\Lambda$ and $\Bo_\Gamma$.
For each $x,y\in\Z_d$, we define an observable $\Co_{x,y}$ by 
\begin{equation}\label{eq:cxy}
\Co_{x,y}(j,k) := U_x^\ast V_y^\ast  \Co(j+x,k+y) V_y U_x \, , \qquad j,k\in\Z_d \, . 
\end{equation}
Using the covariance and invariance properties \eqref{eq:cov-A-unsharp} -- \eqref{eq:cov-B-unsharp} it is straightforward to verify that $\Co_{x,y}$ is a joint observable of $\Ao_\Lambda$ and $\Bo_\Gamma$.

We then define $\widetilde{\Co}$ to be the uniform mixture of all $\Co_{x,y}$, i.e.,
\begin{equation}\label{eq:ctilde}
\widetilde{\Co}(j,k) := \frac{1}{d^2}\sum_{x,y\in\Z_d} \Co_{x,y}(j,k) \, , \qquad j,k\in\Z_d \, . 
\end{equation}
Since every $\Co_{x,y}$ is a joint observable of $\Ao_\Lambda$ and $\Bo_\Gamma$, also $\widetilde{\Co}$ is their joint observable.
A direct calculation, using \eqref{eq:VU}, shows that $\widetilde{\Co}$ is a covariant phase space observable.
\end{proof}

We conclude from Proposition \ref{prop:if-jm-then-cov} that two observables $\Ao_\Lambda$ and $\Bo_\Gamma$ are jointly measurable iff their related probability distributions $\Lambda$ and $\Gamma$ are of the form \eqref{eq:structure1} for some $T\in\sh$. Equations \eqref{eq:structure1} can also be rewritten in a slightly different form. Namely, observe that, if $T\in\sh$, then there exists a unit vector $\phi\in\hh\otimes\hh$ such that $T={\rm tr}_2 [\kb{\phi}{\phi}]$, where ${\rm tr}_2$ is the partial trace with respect to the second factor. 
(A vector state giving $T$ via the partial trace is often called a \emph{purification} of $T$).

Conversely, if $\phi\in\hh\otimes\hh$ is a unit vector, then $T={\rm tr}_2 [\kb{\phi}{\phi}]$ is a state. 
Inserting this form into \eqref{eq:structure1} we obtain
\begin{equation}\label{eq:structure2}
\Lambda (j) = \ip{\phi}{(\Ao (-j)\otimes \id)\phi}  \, , \quad \Gamma (k) = \ip{\phi}{(\Bo(-k)\otimes \id)\phi} \, .
\end{equation}
Note that if a vector $\phi\in\hh\otimes\hh$ satisfies the above two equations for some probability densities $\Lam$ and $\Gamma$, then the normalization $\no{\phi} = 1$ is automatic. 
We thus have the following characterization of jointly measurable observables.

\begin{proposition}\label{prop:joint->phi}
Let $\Lam$, $\Gamma$ be probability densities on $\Z_d$. The following facts are equivalent:
\begin{itemize}
\item[{\rm (i)}] The observables $\Ao_\Lambda$ and $\Bo_\Gamma$ are jointly measurable.
\item[{\rm (ii)}] There exists a state $T \in \sh$ such that the probability densities $\Lambda$ and $\Gamma$ satisfy \eqref{eq:structure1}.
\item[{\rm (iii)}] There exists a vector $\phi \in \hh\otimes\hh$ such that the probability densities $\Lambda$ and $\Gamma$ satisfy \eqref{eq:structure2}.
\end{itemize}
\end{proposition}

Let us note that any $\Ao_\Lambda$ is jointly measurable with some $\Bo_\Gamma$.
Namely, for each $\Lambda$ we can define a state $T_\Lambda$ as
\begin{equation}
T_\Lambda := \sum_{j\in\Z_d} \Lambda(-j) \kb{\varphi_j}{\varphi_j} \, .
\end{equation}
Then $\Lambda (j) = \tr{\Ao (-j)T_\Lambda}$ and $\Ao_\Lambda$ is thus jointly measurable with $\Bo_\Gamma$, where $\Gamma$ is defined as $\Gamma (k) = \tr{\Bo (-k)T_\Lambda}$.

Proposition \ref{prop:joint->phi} can be seen as a trade-off relation between the probability distributions $\Lambda$ and $\Gamma$ that describe the deviations of $\Ao_\Lambda$ and $\Bo_\Gamma$ from $\Ao$ and $\Bo$, respectively.
For instance, if $\Lambda=\delta$, then necessarily $\Gamma=\mu$.
Hence, we recover the fact that $\Ao$ is jointly measurable only with the trivial observable and no other $\Bo_\Gamma$. 

We end this subsection with some additional observations.

\begin{remark}
\emph{If two observables $\Ao_\Lambda$ and $\Bo_\Gamma$ are jointly measurable, they can have several different covariant phase space observables as their joint observables.}

For instance, let $\{\zeta_i\}_{i\in\Z_d}$ be an orthonormal basis which is mutually unbiased with respect to both orthonormal bases $\{\varphi_j\}_{j\in\Z_d}$ and $\{\psi_k\}_{k\in\Z_d}$.
Then, for each $i\in\Z_d$, we have 
\begin{equation*}
\ip{\zeta_i}{\Ao (-j)\zeta_i}=\ip{\zeta_i}{\Bo (-k)\zeta_i}=1/d \, .
\end{equation*}
Therefore, the marginals of the covariant phase space observables $\Co_{\kb{\zeta_i}{\zeta_i}}$ are the same although $\Co_{\kb{\zeta_i}{\zeta_i}}\neq\Co_{\kb{\zeta_{i'}}{\zeta_{i'}}}$ whenever  $i\neq i'$.
 \end{remark}
 
\begin{remark}
\emph{If two observables $\Ao_\Lambda$ and $\Bo_\Gamma$ are jointly measurable, they can have a joint observable which is not a covariant phase space observable.}

For instance, let $p:\Z_d\times\Z_d\to[0,1]$ be a bivariate probability distribution with uniform margins.
Then the observable $\Co(j,k):=p(j,k)\id$ is a joint observable of $\Ao_0$ and $\Bo_0$.
It is clear that $\Co$ is a covariant phase space observable only if $p$ is a uniform distribution.
However, a bivariate probability distribution with uniform margins need not be uniform.
For instance, if we set
\begin{equation*}
p(i,j) = \frac{1}{d^2} \left( 1 - \sin \left( 2\pi \frac{ij}{d} \right) \right) \, , \quad i,j\in\Z_d \, , 
\end{equation*}
then $\sum_i p(i,j)=\sum_j p(i,j)=1/d$, but $p$ is not uniform.

The existence of non-covariant joint observables is not limited to the trivial observables $\Ao_0$ and $\Bo_0$.
Namely, suppose that $\Co$ is a joint observable of $\Ao_0$ and $\Bo_0$ and $\Co'$ is a joint observable of $\Ao_\lambda$ and $\Bo_\gamma$.
Then the convex combination $t\Co' + (1-t)\Co$, $0<t<1$, is a joint observable of $\Ao_{t\lambda}$ and $\Bo_{t\gamma}$.
It is easy to see that if $\Co'$ is a covariant phase space observable but $\Co$ is not, then their convex combination $t\Co' + (1-t)\Co$ cannot be a covariant phase space observable.
\end{remark}

\begin{remark}\label{prop:unique}
\emph{Suppose that $\Ao_\lam$ and $\Bo_\gamma$ have a unique joint observable $\Co_T$ among the covariant phase space observables and that $T^2=T$.
Then $\Co_T$ is a unique joint observable of $\Ao_\lambda$ and $\Bo_\gamma$.}

To prove this claim, let $\Co$ be a joint observable of $\Ao_\lambda$ and $\Bo_\gamma$.
We need to show that $\Co=\Co_T$.
We define $\Co_{x,y}$ and $\widetilde{\Co}$ as in \eqref{eq:cxy} -- \eqref{eq:ctilde}.
Since $\widetilde{\Co}$ is by its construction a covariant phase space observable, we must have $\widetilde{\Co}=\Co_T$ by the assumption on uniqueness.
In particular, each operator $\widetilde{\Co}(j,k)=(1/d)\, U_jV_k T V_k^\ast U_j^\ast$ is rank-1. Since $0\leq \Co_{x,y} (0,0) \leq d^2 \, \widetilde{\Co}(0,0) = d\, T$ by \eqref{eq:ctilde}, it follows that there exists a real constant $0\leq c(x,y) \leq d$ such that $\Co_{x,y} (0,0) = c(x,y) \, T$, hence
$$
\Co(x,y) = c(x,y) \, U_x V_y T V_y^\ast U_x^\ast
$$
by \eqref{eq:cxy}. 
Suppose $\lam \neq 1$. 
Since $\Co$ has $\Ao_\lam$ as its first marginal, then
$$
\sum_{y\in\Z_d} c(x,y) \, U_x V_y T V_y^\ast U_x^\ast = \lam \Ao (x) + (1-\lam) \frac{1}{d} \id \quad \forall x\in\Z_d \, .
$$
The right hand side of this equation is a rank-$d$ operator, while on the left hand side we have the sum of $d$ operators with rank-$1$. 
It then follows that the set $\{ U_x V_y T V_y^\ast U_x^\ast \}_{y\in\Z_d}$ is linearly independent in $\lh$. Since $\Co$ and $\widetilde{\Co}$ have the same marginals,
$$
\sum_{y\in\Z_d} c(x,y) \, U_x V_y T V_y^\ast U_x^\ast = \frac{1}{d} \sum_{y\in\Z_d} U_x V_y T V_y^\ast U_x^\ast
$$
for all $x$, hence $c(x,y) = 1/d$ for all $x,y$ by linear independence. 
The case $\lam = 1$ is treated in a similar way, by taking the marginal $\Bo_\gamma$ in the place of $\Ao_\lam$ (and now necessarily $\gamma\neq 1$).
Therefore, $\Co_T$ is the unique joint observable of $\Ao_\lam$ and $\Bo_\gamma$.
\end{remark}

\subsection{Unsharpness inequality}

In this subsection we apply Proposition \ref{prop:joint->phi} to the cases where $\Ao_\Lambda=\Ao_\lambda$ and $\Bo_\Gamma=\Bo_\gamma$.
These special type of marginal observables are interesting as we can quantify their unsharpnesses by single numbers $\lambda$ and $\gamma$.
In particular, we can ask how small $\lambda$ and $\gamma$ must be in order for  $\Ao_\lambda$ and $\Bo_\gamma$ to become jointly measurable.
We first notice that this question is, indeed, meaningful.

\begin{proposition}\label{prop:ord. makes sense}
Let $\lambda,\gamma\in (0,1]$. The following conditions are equivalent:
\begin{itemize}
\item[(i)] $\Ao_\lambda$ and $\Bo_\gamma$ are jointly measurable.
\item[(ii)] $\Ao_{\lambda'}$ and $\Bo_{\gamma'}$ are jointly measurable for all $0\leq\lambda' \leq \lambda$ and $0\leq\gamma'\leq\gamma$.
\item[(iii)] $\Ao_{\lambda'}$ and $\Bo_{\gamma'}$ are jointly measurable for all $0\leq\lambda' <\lambda$ and $0\leq\gamma' <\gamma$.
\end{itemize}
\end{proposition}

\begin{proof}
Suppose that (i) holds and $0<\gamma' <\gamma$.
We denote $t:=\gamma'/\gamma$ and hence $0<t < 1$.
We have
\begin{equation*}
\Bo_{\gamma'}(k) = t \Bo_{\gamma}(k) + (1-t) \frac{1}{d} \id \, ,
\end{equation*}
meaning that $\Bo_{\gamma'}$ is a convex combination of $\Bo_{\gamma}$ and the trivial observable $\Bo_0$.
By the assumption $\Ao_\lambda$ is jointly measurable with $\Bo_{\gamma}$ and $\Ao_\lambda$ is also jointly measurable with the trivial observable $\Bo_0$ (since they commute).
If $\Co_1$ is a joint observable of $\Ao_\lambda$ and $\Bo_{\gamma}$ and $\Co_2$ is a joint observable of $\Ao_\lambda$ and $\Bo_0$, then the convex combination $t\Co_1 + (1-t)\Co_2$ is a joint observable of $\Ao_\lambda$ and  $\Bo_{\gamma'}$.
Therefore, $\Ao_\lambda$ and  $\Bo_{\gamma'}$ are jointly measurable.
We can interchange the roles of $\Ao_\lambda$ and  $\Bo_{\gamma}$ and run the same argument, hence we obtain (ii).

It is clear that (ii) implies (iii). 
Hence, to complete the proof we need to show that (iii) implies (i). 

Suppose that (iii) holds. 
We choose sequences $(\lambda_n)$ and $(\gamma_n)$ such that $0<\lambda_n<\lambda$, $0<\gamma_n<\gamma$ and $\lim_n \lambda_n=\lambda$, $\lim_n\gamma_n=\gamma$.
For each $n$, we fix a state $T_n$ such that the corresponding covariant phase space observable $\Co_{T_n}$ is a joint observable of $\Ao_{\lambda_n}$ and $\Bo_{\gamma_n}$.
The set of states $\sh$ is compact in the operator norm topology, hence the sequence $(T_n)_n$ has a convergent subsequence. 
We denote by $T$ the limit of this convergent subsequence.
Using \eqref{eq:structure1} we see that the covariant phase space observable $\Co_T$ is a joint observable of $\Ao_\lambda$ and $\Bo_\gamma$.
 Thus, (i) holds.
\end{proof}

We conclude from Proposition \ref{prop:ord. makes sense} that for every $\lambda\in[0,1]$, there is a number $\gamma_{\max}(\lambda) \geq 0$ such that $\Ao_{\lambda}$ and $\Bo_{\gamma}$ are jointly measurable iff $0\leq\gamma\leq\gamma_{\max}(\lambda)$.
Similarly, for every $\gamma\in[0,1]$, there is a number $\lambda_{\max}(\gamma)\geq 0$ such that $\Ao_{\lambda}$ and $\Bo_{\gamma}$ are jointly measurable iff $0\leq\lambda\leq\lambda_{\max}(\gamma)$.

We also know that $\gamma_{\max}(0)=\lambda_{\max}(0)=1$ (since a trivial observable is jointly measurable with any other observable) and that $\gamma_{\max}(1)=\lambda_{\max}(1)=0$ (see the discussion after Proposition \ref{prop:joint->phi}). 

\begin{proposition}
The equality $\lam_{\max} (x) = \gamma_{\max} (x)$ holds for all $x\in[0,1]$.
\end{proposition}

\begin{proof}
It is enough to show that $\Ao_\lam$ and $\Bo_\gamma$ are jointly measurable if and only if $\Ao_\gamma$ and $\Bo_\lam$ are such. 
Joint measurability of $\Ao_\lam$ and $\Bo_\gamma$ means that there exists an observable $\Co$ on $\Z_d \times \Z_d$ having marginals $\Ao_\lam$ and $\Bo_\gamma$, respectively. 
We set $\widehat{\Co} (j,k) := \ff^\ast \Co(k,-j) \ff$ for every $j,k\in\Z_d$. 
Then
\begin{align*}
\sum_{k\in\Z_d} \widehat{\Co} (j,k) & = \sum_{k\in\Z_d} \ff^\ast \Co(k, -j) \ff = \ff^\ast \Bo_\gamma (-j) \ff \\
& = \ff^{\ast \, 2} \Ao_\gamma (-j) \ff^2 = \Ao_\gamma (j)  \, , \\ 
\sum_{j\in\Z_d} \widehat{\Co} (j, k) & = \sum_{j\in\Z_d} \ff^\ast \Co(k , -j) \ff = \ff^\ast \Ao_\lam (k) \ff \\
& = \Bo_\lam (k) \, .
\end{align*}
We conclude that $\widehat{\Co}$ is a joint observable of $\Ao_\gamma$ and $\Bo_\lam$, hence the latter two  are jointly measurable.
\end{proof}

We will now find out the function $\gamma_{\max}(\lambda)$, or, equivalently, $\lam_{\max}(\gamma)$.
Suppose that $\lambda,\gamma\in [0,1]$ are such that $\Ao_\lambda$ and $\Bo_\gamma$ are jointly measurable.
By Proposition \ref{prop:joint->phi} this is equivalent to the existence of a vector $\phi \in \hh\otimes\hh$ satisfying
\begin{eqnarray}
\ip{\phi}{(\Ao (j)\otimes \id)\phi} & = & \lam \delta (j) + (1-\lam) \mu (j) \, ,\label{eq. 1} \\
\ip{\phi}{(\Bo(k)\otimes \id)\phi} & = & \gamma \delta (k) + (1-\gamma) \mu (k)\, ,\label{eq. 2}
\end{eqnarray}
for all $j,k\in\Z_d$ (see \eqref{eq:point-and-uniform} for the definition of $\delta$ and $\mu$).
We now give a condition on the parameters $\lam,\gamma$ which is necessary and sufficient for the existence of a vector $\phi \in \hh\otimes\hh$ satisfying the above two equations. Moreover, we show that, for the extreme values of $\lam,\gamma$, the vector $\phi$ is essentially unique.

\begin{lemma}\label{lemma:bound}
Let $\lambda,\gamma\in [0,1]$. 
Suppose there exists a vector $\phi\in\hh\otimes\hh$ satisfying \eqref{eq. 1} -- \eqref{eq. 2}. 
Then 
\begin{equation}\label{ineq. gamma}
\gamma \leq \frac{1}{d} \left[ (d-2)(1-\lam) + 2 \sqrt{(1-d)\lam^2 + (d-2)\lam + 1} \right] \, .
\end{equation}

For any choice of a unit vector $\eta \in \hh$, the vector
\begin{equation}\label{phi est.}
\phi = (\alpha_\lam \fii_0 + \beta_\lam \psi_0) \otimes \eta \, ,
\end{equation}
with
\begin{align*}
\alpha_\lam & = \frac{1}{\sqrt{d}} \left[\sqrt{(d-1)\lam + 1} - \sqrt{1-\lam}\right] \, , \qquad
\beta_\lam  =  \sqrt{1-\lam} \, ,
\end{align*}
satisfies \eqref{eq. 1} -- \eqref{eq. 2} with equality in \eqref{ineq. gamma}. 
Hence, the right hand side in \eqref{ineq. gamma} is equal to $\gamma_{\max}(\lambda)$.

If $\phi'$ is a vector satisfying \eqref{eq. 1} -- \eqref{eq. 2} with  $\gamma=\gamma_{\max}(\lambda)$, then $\phi' = (\alpha_\lam \fii_0 + \beta_\lam \psi_0) \otimes \eta'$ for some unit vector $\eta' \in \hh$.
\end{lemma}

\begin{proof}
Suppose $\phi\in\hh\otimes\hh$ satisfies \eqref{eq. 1} -- \eqref{eq. 2}. 
We write $\phi$ in the form
$$
\phi = \sum_{i\in\Z_d} \fii_i \otimes \xi_i \, ,
$$
where $\{ \xi_i\}_{i\in\Z_d}$ are vectors in $\hh$. 
From \eqref{eq. 1} it follows that
\begin{equation*}
\no{\xi_i}^2 = \lam \delta (i) + (1-\lam) \mu (i) \, ,
\end{equation*}
hence there exist unit vectors $\{ \eta_i\}_{i\in\Z_d}$ such that
$$
\xi_0 = \sqrt{\frac{(d-1)\lam +1}{d}} \, \eta_0 \, , \qquad \xi_i = \sqrt{\frac{1-\lam}{d}} \, \eta_i \quad \forall i\neq 0 \, .
$$
On the other hand, we have
\begin{eqnarray*}
\ip{\phi}{(\Bo(k)\otimes \id)\phi} & = & \sum_{i,j\in\Z_d} \ip{\fii_j\otimes\xi_j}{(\kb{\psi_k}{\psi_k} \otimes \id) \fii_i\otimes\xi_i} \\
& = & \frac{1}{d} \sum_{i,j\in\Z_d} \omega^{jk} \omega^{-ik} \ip{\xi_j}{\xi_i} \, ,
\end{eqnarray*}
so, by \eqref{eq. 2}, we must have
$$
\frac{1}{d} \sum_{i,j\in\Z_d} \omega^{jk} \omega^{-ik} \ip{\xi_j}{\xi_i} = \gamma \delta (k) + (1-\gamma) \mu (k) \, .
$$
This equation, evaluated at $k=0$, gives
\begin{eqnarray*}
\left( 1 - \frac{1}{d} \right) \gamma + \frac{1}{d} & = & \frac{1}{d} \no{\sum_{i\in\Z_d} \xi_i}^2 \\
& = & \frac{1}{d^2} \no{\sqrt{(d-1)\lam +1} \, \eta_0 + \sum_{i\in\Z_d, \, i\neq 0} \sqrt{1-\lam} \, \eta_i}^2 \, .
\end{eqnarray*}
The maximum value of $\gamma$ is then achieved when the right hand side of this equation is maximal, i.e., when there exists a unit vector $\eta\in\hh$ such that $\eta_i = \eta$ $\forall i\in\Z_d$. The corresponding maximum value $\gamma_{\max}$ of $\gamma$ is given by
$$
\left( 1 - \frac{1}{d} \right) \gamma_{\max} + \frac{1}{d} = \frac{1}{d^2} \left( \sqrt{(d-1)\lam +1} + (d-1) \sqrt{1-\lam} \right)^2 \, ,
$$
i.e.,
\begin{equation*}
\gamma_{\max} = \frac{1}{d} \left[ (d-2)(1-\lam) + 2 \sqrt{(1-d)\lam^2 + (d-2)\lam + 1} \right] \, .
\end{equation*}

In order to show that, if the sequence $\{\xi_i\}_{i\in\Z_d}$ is chosen as above, then the corresponding vector
\begin{eqnarray*}
\phi & = & \left( \sqrt{\frac{(d-1)\lam +1}{d}} \, \fii_0 + \sqrt{\frac{1-\lam}{d}} \sum_{i\in\Z_d, \, i\neq 0} \fii_i \right) \otimes \eta \\
& = & (\alpha_\lam \fii_0 + \beta_\lam \psi_0) \otimes \eta
\end{eqnarray*}
satisfies also \eqref{eq. 2} with $\gamma=\gamma_{\max}$ (and thus the maximum is indeed achieved by $\phi$), we evaluate
\begin{eqnarray*}
\ip{\phi}{(\Bo(k)\otimes \id)\phi} & = & \ip{\phi}{(\kb{\psi_k}{\psi_k} \otimes \id)\phi} \\
& = & \left|\ip{\psi_k}{\alpha_\lam \fii_0 + \beta_\lam \psi_0}\right|^2 \\
& = & \left( \frac{\alpha_\lam}{\sqrt{d}} + \beta_\lam \delta(k)\right)^2 \\
& = & \frac{\alpha^2_\lam}{d} + \left(\beta^2_\lam + 2 \frac{\alpha_\lam \beta_\lam}{\sqrt{d}} \right) \delta(k) \\
& = & (1-\gamma_{\max}) \mu (k) + \gamma_{\max} \delta (k) \, ,
\end{eqnarray*}
which is \eqref{eq. 2}.
\end{proof}

As a consequence of the above discussion, we obtain an inequality for the unsharpnesses of two jointly measurable observables $\Ao_\lambda$ and $\Bo_\gamma$.

\begin{figure}
\begin{center}
\includegraphics[width=4.0cm]{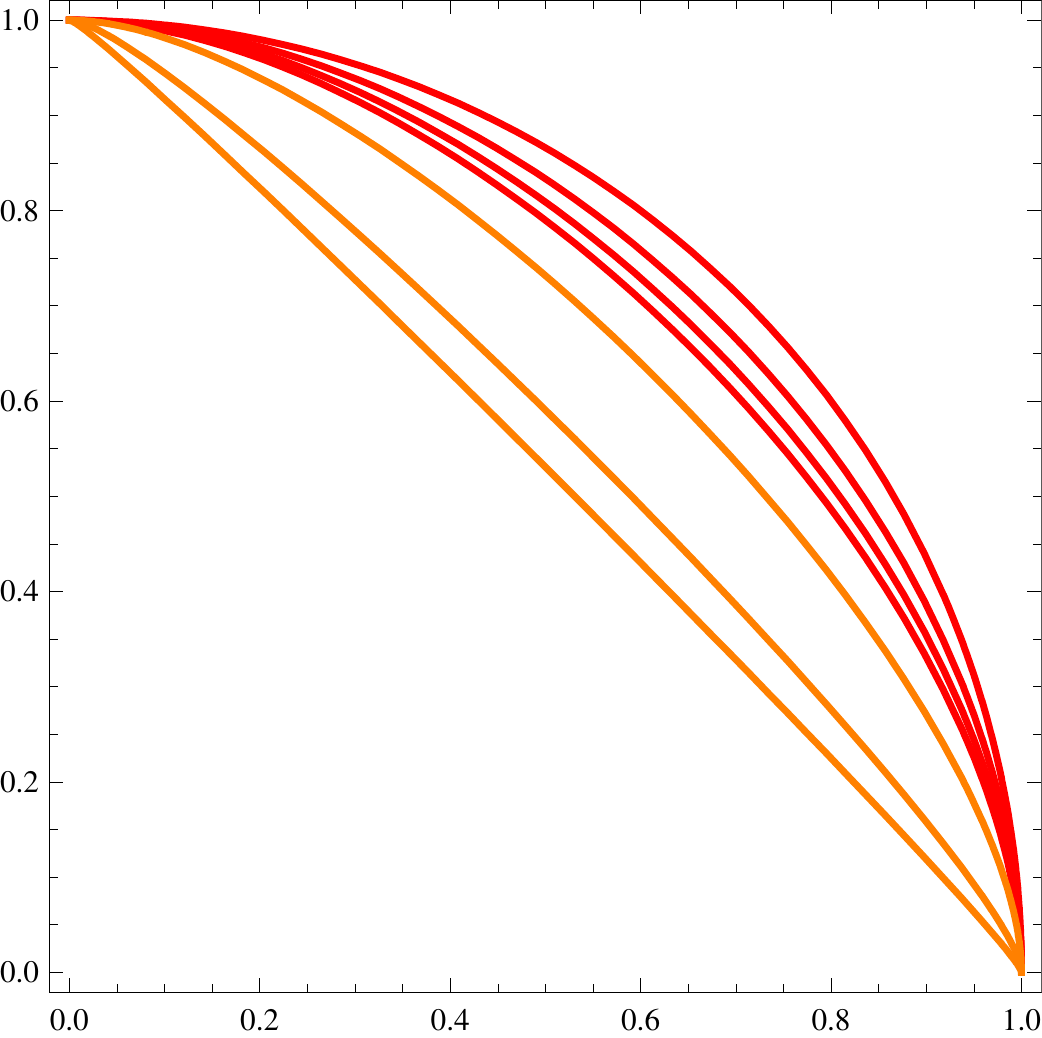}
\end{center}
\caption{The boundary curve $\lambda\mapsto\gamma_{\max}(\lambda)$ for $d=2,3,4,5$ (red color) and for $d=10,100,1000$ (orange color).}
\label{fig:boundary}
\end{figure}

\begin{proposition}\label{Prop:gammamax}
Two observables $\Ao_\lambda$ and $\Bo_\gamma$ are jointly measurable if and only if 
\begin{equation}\label{ineq. gamma 2}
\gamma \leq \gamma_{\max}(\lambda)=\frac{1}{d} \left[ (d-2)(1-\lam) + 2 \sqrt{(1-d)\lam^2 + (d-2)\lam + 1} \right]  \, ,
\end{equation}
(or, equivalently, its modified form under the exchange $\gamma\leftrightarrow\lam$). 

If $\gamma=\gamma_{\max}(\lambda)$, then $\Ao_\lambda$ and $\Bo_\gamma$ have a unique joint observable.
This unique joint observable is the covariant phase space observable $\Co_T$ defined by the state
$$
T = \kb{\chi_\lam}{\chi_\lam} \, , \qquad \chi_\lam = \alpha_\lam \fii_0 + \beta_\lam \psi_0 \, ,
$$
with
\begin{align}\label{eq:ab}
\alpha_\lambda &=  \frac{1}{\sqrt{d}} \left[ \sqrt{(d-1) \lambda + 1} - \sqrt{1-\lambda} \right] \, , \qquad
\beta_\lambda  =  \sqrt{1-\lambda} \, .
\end{align}
\end{proposition}

\begin{proof}
If $\Ao_\lam$ and $\Bo_\gamma$ are jointly measurable, then the inequality follows from Proposition \ref{prop:joint->phi} and Lemma \ref{lemma:bound}. 
Conversely, if $\gamma_{\max}$ is given by the right hand side of \eqref{ineq. gamma 2}, then the pair $\Ao_\lam$ and $\Bo_{\gamma_{\max}}$ are jointly measurable again by an application of Proposition \ref{prop:joint->phi} and Lemma \ref{lemma:bound}. 
Then, $\Ao_\lam$ and $\Bo_\gamma$ are jointly measurable by Proposition \ref{prop:ord. makes sense}.

Now suppose $\lam$ and $\gamma$ achieve the bound \eqref{ineq. gamma 2}, and let $\Co_T$ be a covariant joint observable of $\Ao_\lam$ and $\Bo_\gamma$. 
Pick $\phi\in\hh\otimes\hh$ such that $T={\rm tr}_2 [\kb{\phi}{\phi}]$. 
As $\phi$ satisfies \eqref{eq. 1} -- \eqref{eq. 2} with $\gamma = \gamma_{\max}$, by Lemma \ref{lemma:bound} it must be given by \eqref{phi est.} for some choice of a unit vector $\eta \in \hh$, and $T={\rm tr}_2 [\kb{\phi}{\phi}] = \kb{\chi_\lam}{\chi_\lam}$, with $\chi_\lam$ as in the statement of the proposition.

Finally, we need to prove that $\Co_T$ is the unique joint observable (and not only unique among covariant phase space observables). 
We notice that $T^2=T$, and the claim thus follows from Remark \ref{prop:unique}.
\end{proof}

The graph of the function $\lambda\mapsto\gamma_{\max}(\lambda)$ is a part of an ellipse.
In Fig.~\ref{fig:boundary} we have depicted it for $d=2,3,4,5,10,100,1000$.

\begin{example}
Suppose that two jointly measurable observables $\Ao_\lambda$ and $\Bo_\gamma$ are `equally unsharp' but as close to $\Ao$ and $\Bo$ as possible, i.e., 
\begin{equation*}
\gamma=\lambda = \lambda_{\max}(\gamma) \, .
\end{equation*}
In this case Proposition \ref{Prop:gammamax} gives
\begin{equation}
\gamma=\lambda = \frac{d+\sqrt{d}-2}{2(d-1)} = \half \left( 1+ \frac{1}{1+\sqrt{d}} \right) \, .
\end{equation}
The observables $\Ao_\lambda$ and $\Bo_\gamma$ then have a unique joint observable, which is the covariant phase space observable $\Co_T$ associated to the state $T=\kb{\chi}{\chi}$, with
$$
\chi =\sqrt{\frac{\sqrt{d}}{2(1+\sqrt{d})}} \, (\fii_0 + \psi_0) \, .
$$
The vector state $\chi$ is hence an equal superposition of the vector states $\fii_0$ and $\psi_0$.
\end{example}

By a direct calculation one can verify that the inequality \eqref{ineq. gamma 2} can be rewritten in the following equivalent form which is symmetric in $\lambda$ and $\gamma$.

\begin{proposition}\label{prop:symmetric}
Two observables $\Ao_\lambda$ and $\Bo_\gamma$ are jointly measurable iff
\begin{equation}\label{eq:region}
\textrm{either} \quad \gamma+\lambda\leq 1 \quad \textrm{or} \quad \gamma^2 + \lambda^2 +  \frac{2(d-2)}{d} (1- \gamma)(1- \lambda)  \leq  1 \, .
\end{equation}
\end{proposition}

The second inequality in \eqref{eq:region} describes a full ellipse. 
Therefore, the first condition $\gamma+\lambda\leq 1$ is needed to ignore the lower part of the ellipse, which is not a correct boundary for joint measurability.
This is depicted in Fig.~\ref{fig:ellipse}.

\begin{figure}
\begin{center}
\includegraphics[width=4.0cm]{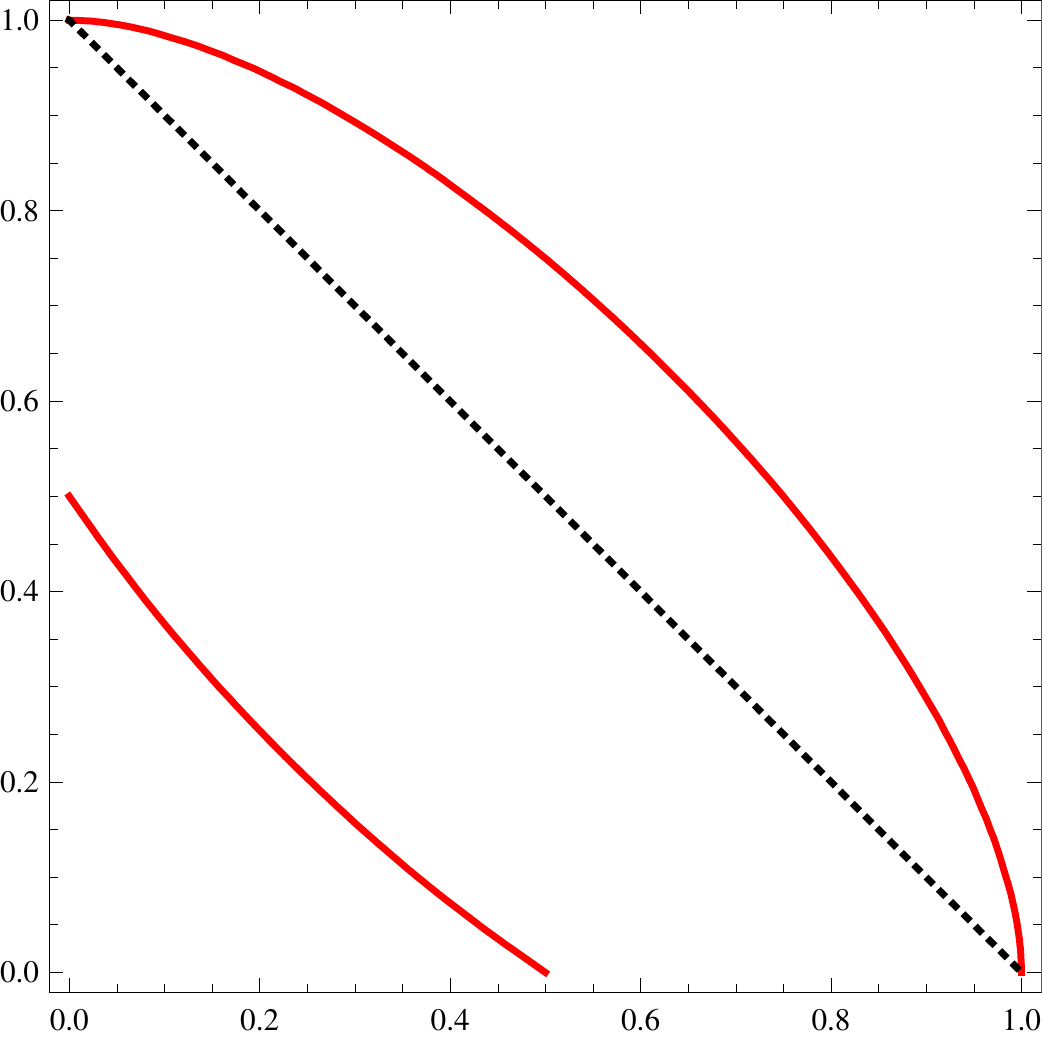}
\end{center}
\caption{The ellipse (red) and the line (dashed) from \eqref{eq:region} for $d=8$. 
Only the upper side of the ellipse is relevant for the joint measurability.}
\label{fig:ellipse}
\end{figure}

Let us notice that for $d=2$ the latter inequality in \eqref{eq:region} becomes $\gamma^2 + \lambda^2\leq 1$, and then the first condition is redundant.
We thus recover the single condition stated in \eqref{eq:ur-qubit-2} and first proved in \cite{Busch86}.
Also, for $d=3$ and $d=4$ a direct calculation shows that the the first condition in \eqref{eq:region} is redundant and the latter quadratic inequality is necessary and sufficient for the joint measurability.
For $d\geq 5$ we need both conditions in \eqref{eq:region}.

By inspecting the function $\lambda\mapsto\gamma_{\max}(\lambda)$ we see that for every $\epsilon>0$, there is a pair of observables $\Ao_\lambda$ and $\Bo_\gamma$ such that they are not jointly measurable and $\lambda + \gamma < 1+ \epsilon$.
Thus, the criterion $\gamma+\lambda\leq 1$ is the best \emph{sufficient} condition for joint measurability which is linear and symmetric in $\lambda$ and $\gamma$.

The best linear and symmetric \emph{necessary} condition for joint measurablity is achived by taking the tangent of the boundary curve in the point where it crosses the line $\gamma=\lambda$.
In this way, we obtain the following conclusion.

\begin{proposition}\label{prop:linear}
If $\Ao_\lambda$ and $\Bo_\gamma$ are jointly measurable, then
\begin{equation}
\gamma+\lambda \leq 1+ \frac{\sqrt{d}-1}{d-1} \, .
\end{equation}
\end{proposition}

One can also see Proposition \ref{prop:linear} in the opposite order; if $\gamma+\lambda > 1+ \frac{\sqrt{d}-1}{d-1}$, then $\Ao_\lambda$ and $\Bo_\gamma$ are not jointly measurable.
In Fig.~\ref{fig:lines} we have depicted the linear necessary and sufficient conditions in the case $d=10$.

\begin{figure}
\begin{center}
\includegraphics[width=4.0cm]{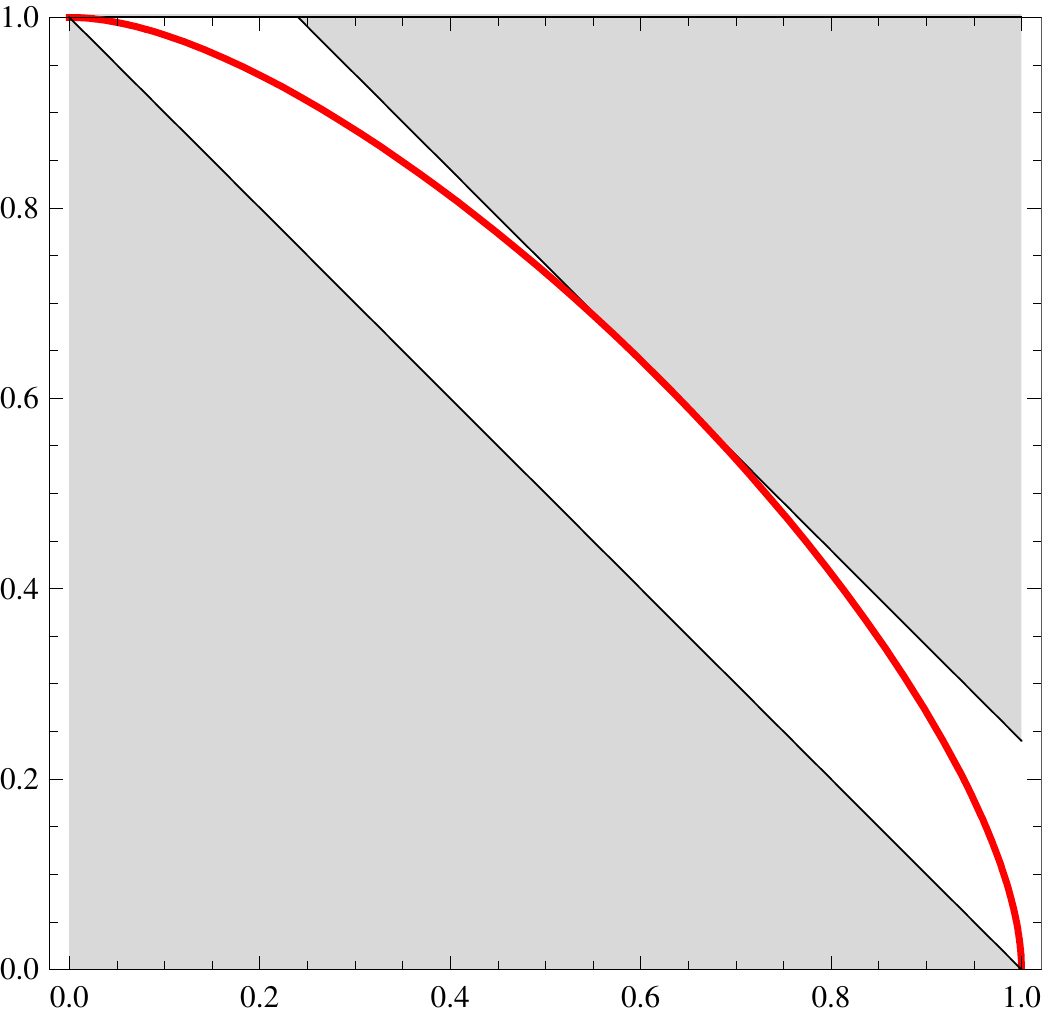}
\end{center}
\caption{In this picture $d=10$. The gray regions represent the necessary and sufficient linear conditions. In the white stripe one has to invoke the quadratic ellipse criterion, whereas otherwise the joint measurability can be deduced from the simple linear criteria.}
\label{fig:lines}
\end{figure}

\subsection{Non-covariant observables}

So far, we have concentrated on covariant observables $\Ao_\Lambda$ and $\Bo_\Gamma$. 
Let us have a short look on a class of non-covariant observables.

Let $p$ and $r$ be two probability distributions on $\Z_d$.
We define
\begin{equation*}
\Ao_{\lambda;p}(j) := \lambda \Ao(j) + (1-\lambda) p(j) \id 
\end{equation*}
and
\begin{equation*}
\Bo_{\gamma;r}(k) := \gamma \Bo(k) + (1-\gamma) r(k) \id \, .
\end{equation*}
It is straightforward to verify that $\Ao_{\lambda;p}$ is $U$-covariant and $V$-invariant iff $p$ is the uniform distribution on $\Z_d$, in which case $\Ao_{\lambda;p}=\Ao_{\lambda}$.
(Analogous statement holds for $\Bo_{\gamma;r}$).
The following result is a generalization of Proposition 5 in \cite{BuHe08}.

\begin{proposition}\label{prop:non}
If $\Ao_{\lambda;p}$ and $\Bo_{\gamma;r}$ are jointly measurable, then $\Ao_{\lambda}$ and $\Bo_{\gamma}$ are jointly measurable.
\end{proposition}

\begin{proof}
Suppose that $\Ao_{\lambda;p}$ and $\Bo_{\gamma;r}$ are jointly measurable and let $\Co$ be their joint observable. As in the proof of Proposition \ref{prop:if-jm-then-cov}, we define the observable $\widetilde{\Co}$, given by
$$
\widetilde{\Co}(j,k) := \frac{1}{d^2}\sum_{x,y\in\Z_d} U_x^\ast V_y^\ast  \Co(j+x,k+y) V_y U_x \, , \qquad j,k\in\Z_d \, .
$$
For each $j$, we have
\begin{eqnarray*}
\sum_{k\in\Z_d} \widetilde{\Co}(j,k) &=& \frac{1}{d^2}  \sum_{x,y\in\Z_d} U_x^\ast V_y^\ast  \sum_{k\in\Z_d} \Co(j+x,k+y) V_y U_x \\
&=&  \frac{1}{d^2}  \sum_{x,y\in\Z_d} U_x^\ast V_y^\ast  \Ao_{\lambda;p}(j+x) V_y U_x \\
&=& \lambda \Ao(j) + (1-\lambda) \frac{1}{d^2}  \sum_{x,y\in\Z_d} p(j+x) \id \\
&=& \Ao_{\lambda}(j) \, .
\end{eqnarray*}
In a similar way we obtain $\sum_{j\in\Z_d} \widetilde{\Co}(j,k) = \Bo_{\gamma}(k)$ for every $k$.
Therefore, $\widetilde{\Co}$ is a joint observable for $\Ao_{\lambda}$ and $\Bo_{\gamma}$.
\end{proof}

As a consequence of Propositions \ref{Prop:gammamax} and \ref{prop:non} we conclude the following necessary criterion for joint measurability.

\begin{corollary}\label{cor:nec}
If two observables $\Ao_{\lambda;p}$ and $\Bo_{\gamma;r}$ are jointly measurable, then 
\begin{equation}
\gamma \leq \frac{1}{d} \left[ (d-2)(1-\lam) + 2 \sqrt{(1-d)\lam^2 + (d-2)\lam + 1} \right]  \, .
\end{equation}
\end{corollary}

A necessary and sufficient inequality for the joint measurability of $\Ao_{\lambda;p}$ and $\Bo_{\gamma;r}$ must contain also $p$ and $r$ in a form or in another.
It is thus clear that Corollary \ref{cor:nec} does not give a sufficient condition.
A necessary and sufficient condition in the case $d=2$ has been obtained in \cite{StReHe08, BuSc10, YuLiLiOh10}.

We remark that a general necessary condition for the joint measurability of two observables on a finite dimensional system has been presented in \cite{MiIm08}.
A comparison to Proposition \ref{Prop:gammamax} shows that this condition is not sufficient.  
We leave it as an open problem to find a necessary and sufficient condition for the joint measurability of $\Ao_{\lambda;p}$ and $\Bo_{\gamma;r}$.

\section{Informational completeness}\label{sec:ic}

We will now study the informational completeness of joint observables of $\Ao_\lam$ and $\Bo_\gamma$.
Let us first recall that the informational completeness of a covariant phase space observable $\Co_T$ is equivalent to the criterion
\begin{equation}\label{eq:ic}
\tr{TU_x V_y} \neq 0 \quad \forall x,y\in\Z_d \, .
\end{equation}
This result has been discussed e.g. in \cite{AlPr77JMP,CaDeLaLe00JMP,DaPeSa04}.
For completeness, we provide a proof in Appendix.

\begin{proposition}\label{prop:icmax}
Suppose $\Ao_\lam$ and $\Bo_\gamma$ are two observables with $\lambda \notin \{0,1\}$ and $\gamma =\gamma_{\max}(\lambda)$.
Then $\Ao_\lam$ and $\Bo_\gamma$ have a unique joint observable $\Co$.
The observable $\Co$ is informationally complete if and only if $d$ is odd.
\end{proposition}

\begin{proof}
From Proposition \ref{Prop:gammamax} we know $\Ao_\lambda$ and $\Bo_\gamma$ have a unique joint observable $\Co_T$, generated by the state $T = \kb{\chi_\lam}{\chi_\lam}$, with $\chi_\lam = \alpha_\lam \fii_0 + \beta_\lam \psi_0$. 
The informational completeness of $\Co_T$ is equivalent to the condition \eqref{eq:ic}, and 
a straightforward calculation gives
\begin{equation*}
\tr{TU_x V_y} = \ip{\chi_\lam}{U_xV_y\chi_\lam}= \alpha_\lam^2 \delta_{x,0} + \beta_\lam^2 \delta_{y,0} + \frac{\alpha_\lam \beta_\lam}{\sqrt{d}} \left( \omega^{-xy} + 1 \right) \, .
\end{equation*}
Let us first notice that $\alpha_\lam > 0$ and $\beta_\lambda > 0$ since both $\lambda$ and $\gamma$ are nonzero (see \eqref{eq:ab}).
Hence, $\tr{TU_x V_y}=0$ exactly when $\omega^{-xy}=-1$.
The latter condition is equivalent to $2xy\equiv d \mod 2d$. 
We conclude that the informational completeness of $\Co_T$ is equivalent to the fact that the equation $2x = d \mod 2d$ has no solution $x\in\Z_d$, and this holds if and only if $d$ is odd.
\end{proof}

In Proposition \ref{prop:icmax} the crucial assumption is that $\gamma =\gamma_{\max}(\lambda)$.
This guarantees that $\Ao_\lambda$ and $\Bo_\gamma$ have a unique joint observable.
If we have $0<\gamma < \gamma_{\max}(\lambda)$, then $\Ao_\lambda$ and $\Bo_\gamma$ have infinitely many joint observables.
In this case it is always possible to choose an informationally complete joint observable, as we prove in the following.

\begin{proposition}\label{prop:icnonmax1}
Suppose $\Ao_\lam$ and $\Bo_\gamma$ are two observables with $\lambda \notin \{0,1\}$ and $0<\gamma <\gamma_{\max}(\lambda)$.
Then they have an informationally complete covariant joint observable.
\end{proposition}

\begin{proof}
Let $(\gamma_0 , \lam_0)$ be the intersection of the half line $\R_+ (\gamma , \lam)$ with the boundary of the domain \eqref{eq:region} in $\R_+^2$, and let $t_0>1$ such that $t_0 (\gamma , \lam) = (\gamma_0 , \lam_0)$. Let $\tau = 1 - 1/t_0 \in (0,1)$. We treat separately the cases of odd and even $d$.

1) Suppose that $d=2n$ is even.
For all $k\in\Z_d$, we denote
\begin{equation*}
X_k := \frac{i}{2} \left( \kb{\fii_{-k}}{\fii_0} - \kb{\fii_0}{\fii_{-k}} + \kb{\fii_k}{\fii_0} - \kb{\fii_0}{\fii_k} \right) \, .
\end{equation*}
The linear maps $X_k$ are selfadjoint trace zero operators for every $k$, and it is easy to check that
$$
\sum_{x\in\Z_d} U_x V_y X_k V_y^* U_x^* = 0 \, , \qquad \sum_{y\in\Z_d} U_x V_y X_k V_y^* U_x^* = 0 \, .
$$
We introduce the selfadjoint operators
$$
X = \sum_{k\in\Z_d} X_k \, ,
$$
and, for $\kappa > 0$,
$$
S_\kappa = \frac{1}{d} \id + \kappa X \, .
$$
If $\kappa < 1/(d\no{X})$, then $S_\kappa\in\sh$. Moreover, the associated covariant phase space observable $\Co_{S_{\kappa}}$ has trivial marginals $\Ao_0$ and $\Bo_0$. 
A straightforward calculation gives
$$
\tr{S_\kappa U_x V_y} = \delta_{x,0}\delta_{y,0} + i \kappa (\omega^{-xy} - 1) \, .
$$
The covariant phase space observable $\Co_{T_\kappa}$ associated to the state $T_\kappa = (1-\tau) \kb{\chi_{\lam_0}}{\chi_{\lam_0}} + \tau S_\kappa$, with $\chi_{\lam_0} = \alpha_{\lam_0} \fii_0 + \beta_{\lam_0} \psi_0$, has margins $\Ao_\lam$ and $\Bo_\gamma$. Moreover,
\begin{eqnarray}\label{eq:infocom}
\tr{T_\kappa U_x V_y} & = & (1-\tau) \lft \alpha_{\lam_0}^2 \delta_{x,0} + \beta_{\lam_0}^2 \delta_{y,0} \rgt + \tau \delta_{x,0} \delta_{y,0} \notag\\
&& + (1-\tau) \frac{\alpha_{\lam_0} \beta_{\lam_0}}{\sqrt{d}} (\omega^{-xy} + 1) + i\kappa \tau (\omega^{-xy} - 1) \, .
\end{eqnarray}
Let
$$
\varepsilon = \min_{\{k\in\Z_d \, , \, k\neq n\}} |\omega^k + 1| \, , \qquad \delta = \max_{\{k\in\Z_d\}} |\omega^k + 1| \, .
$$
For $\kappa < \min\left\{ \alpha_{\lam_0} \beta_{\lam_0} (1-\tau) \varepsilon / (\tau \delta\sqrt{d}) \, , \, 1/(d\no{X}) \right\}$, the right hand side of \eqref{eq:infocom} is nonzero for all $x,y\in\Z_d$, which proves informational completeness of $\Co_{T_\kappa}$ by the criterion \eqref{eq:ic}.

2) Suppose that $d$ is odd. Then, for $T = (1-\tau) \kb{\chi_{\lam_0}}{\chi_{\lam_0}} + (\tau/d) \id$, the associated covariant phase space observable $\Co_T$ has margins $\Ao_\lam$ and $\Bo_\gamma$, and
$$
\tr{TU_x V_y} = (1-\tau) \left[ \alpha_{\lam_0}^2 \delta_{x,0} + \beta_{\lam_0}^2 \delta_{y,0} + \frac{\alpha_{\lam_0} \beta_{\lam_0}}{\sqrt{d}} \left( \omega^{-xy} + 1 \right) \right] + \tau \delta_{x,0} \delta_{y,0} \, ,
$$
which is nonzero for all $x,y\in \Z_d$. 
The informational completeness of $\Co_T$ then follows from the criterion \eqref{eq:ic}.
\end{proof}

The two trivial cases $\lam=0$ or $\gamma=0$ are not very interesting, but for completeness we make  the following observation. 

\begin{proposition}\label{prop:icnonmax2}
Suppose $\Ao_\lam$ and $\Bo_\gamma$ are two observables with $\lam = 0$ or $\gamma=0$. Then they have no informationally complete joint observable.
\end{proposition}

\begin{proof}
We consider only the case $\lam=0$, the case $\gamma=0$ being similar. 
Suppose that $\Co$ is a joint observable of $\Ao_0 = \mu \id$ and $\Bo_\lam$. 
We have
$$
\sum_{k\in\Z_d} \Co(j,k) = \Ao_0 (j) = \frac{1}{d} \id \qquad \forall j\in\Z_d \, ,
$$
hence
$$
{\rm span}\, \left\{ \Co(j,k) \mid j,k\in\Z_d \right\}  = {\rm span}\, \left\{ \id \, , \, \Co(j,k) \mid j\in\Z_d \, , \, k\in\Z_d \setminus \{0\} \right\}
$$
and then, for $d\geq 2$,
$$
\dim {\rm span}\, \left\{ \Co(j,k) \mid j,k\in\Z_d \right\} \leq 1 + d(d-1) < d^2 \, .
$$
Thus, $\Co$ is not informationally complete.
\end{proof}

\section{Sequential implementation of joint observables}\label{sec:seq}

In this section we discuss the sequential implementation of joint observables of $\Ao_\Lambda$ and $\Bo_\Gamma$ in the light of the recent results obtained in \cite{CaHeTo11} and \cite{HeWo10}.  
For illustrative purposes, we point out that two naive methods do not work. 

\subsection{Nondisturbing measurement}

Suppose that $\Ao_\Lambda$ and $\Bo_\Gamma$ are jointly measurable, i.e., they satisfy the condition stated in Proposition \ref{prop:joint->phi}.
The most uncomplicated way to realize their joint measurement would be to perform an $\Ao_\Lambda$-measurement without disturbing the subsequent $\Bo_\Gamma$-measurement.
In terms of instruments, this would mean that we choose an $\Ao_\Lambda$-compatible instrument $\Ii$ such that
\begin{equation*}
\sum_j \Ii_j^\ast(\Bo_\Gamma(k)) = \Bo_\Gamma(k)
\end{equation*}
for all $k$.
However, this type of measurement is typically not possible since a quantum measurement necessarily disturbs the input state.

Let us first notice that $\Ao_\lambda$ and $\Bo_\gamma$ commute if and only if $\lambda\gamma=0$, meaning that one of them is a trivial observable.
Generally, a non-disturbing measurement can be possible even if two observables do not commute.
But applying Proposition 3 from \cite{HeWo10} we see that this possibility is excluded whenever $\Bo_\gamma$ is informationally equivalent with $\Bo$ in the sense that the linear spans of the sets $\{\Bo_\gamma(k):k\in\Z_d\}$ and $\{\Bo(k):k\in\Z_d\}$ are equal.
This property is satisfied by any $\Bo_\gamma$ with $\gamma\neq 0$.
Therefore, whenever both observables are nontrivial, then $\Ao_\lambda$-measurement disturbs the subsequent $\Bo_\gamma$-measurement and the resulting observable is not a joint measurement of $\Ao_\lambda$ and $\Bo_\gamma$.

\subsection{Measuring only part of the ensemble}

Suppose we have a measurement setup for $\Ao$ and that the corresponding instrument is $\Ii$.
We can implement an unsharp observable $\Ao_\lambda$ by performing the $\Ao$-measurement in a randomly chosen $\lambda$-part of the ensemble and doing nothing for the rest $(1-\lambda)$-part.
The corresponding $\Ao_\lambda$-compatible instrument $\Ii'$ is then
\begin{equation}\label{eq:doing-nothing}
\Ii'_j(\varrho) = \lambda \Ii_j(\varrho) + \frac{1-\lambda}{d} \, \varrho \, .
\end{equation}
This is clearly a very direct way to decrease the disturbance that an $\Ao$-measurement would cause. 
By measuring the observable $\Bo$ after the first measurement, one could expect to have a useful joint measurement of $\Ao_\lambda$ and some approximate version of $\Bo$.
However, this type of method does not yield an informationally complete joint measurement.

The observable $\Ao$ consists of rank-1 operators, and any $\Ao$-compatible instrument $\Ii$ is of the form
\begin{equation}
\Ii_j(\varrho)= \tr{\varrho \Ao(j)} \xi_j
\end{equation}
for some set of states $\{\xi_j: j\in\Z_d\}$ \cite{HeWo10}. 
If we insert this form into \eqref{eq:doing-nothing}, we see that a sequential measurement consisting of $\Ii'$ followed by a $\Bo$-measurement leads to the joint observable
\begin{equation*}
\Co(j,k) = \lambda\  \tr{\xi_j \Bo(k)} \Ao(j) + \frac{1-\lambda}{d} \, \Bo(k) \, , \qquad j,k\in\Z_d \, .
\end{equation*}

The linear span of the set $\{\Co(j,k):j,k\in\Z_d \}$ is contained in the linear span of the union $\{\Ao(j): j\in\Z_d\}\cup\{\Bo(k): k\in\Z_d\}$.
The latter is strictly smaller than $\lh$, hence $\Co$ is not informationally complete.
We also see that this kind of approach cannot give more information than separate measurements of $\Ao$ and $\Bo$ would give.

\subsection{General joint observables}

We recall from \cite{HeWo10} that every joint observable of $\Ao_\Lambda$ and $\Bo_\Gamma$ can be implemented as a sequential measurement of $\Ao_\Lambda$ followed by a measurement of $\Bo$.
Namely, suppose that $\Co$ is a joint observable of $\Ao_\Lambda$ and $\Bo_\Gamma$.
We define an instrument $\Ii$ by
\begin{equation}\label{eq:instrument-general}
\Ii_j(\varrho) = \sum_{k\in\Z_d} \tr{\varrho \Co(j,k)} \Bo(k) \, .
\end{equation}
This is an $\Ao_\Lambda$-compatible instrument, and from $\Bo(k)\Bo(k')=\delta_{k,k'}\Bo(k)$ it follows that
\begin{equation*}
\tr{\Bo(k)\Ii_j(\varrho)} = \tr{\varrho \Co(j,k)}  \, .
\end{equation*}
Hence, $\Co(j,k) = \Ii_j^\ast (\Bo(k))$, and we conclude that $\Co$ is implemented as a sequential measurement of $\Ao_\Lambda$ followed by a measurement of $\Bo$, as claimed.

\subsection{Covariant phase space observables}

The instrument defined in \eqref{eq:instrument-general} may look quite artificial and before we know the structure of $\Co$, the formula does not give us any hint on the structure of $\Ii$.
In contrast, every covariant phase space observable can implemented as sequential measurement of $\Ao_\Lambda$ and $\Bo$ in a very specific form.

As explained in \cite{CaHeTo11}, every covariant $\Ao_\Lambda$-compatible instrument gives rise to a covariant phase space observable.
Covariance of an instrument $\Ii$ here means that
\begin{equation}\label{eq:ins-cov}
U_xV_y\Ii_j(V_y^\ast U_x^\ast \varrho U_x V_y)V_y^\ast U_x^\ast = \Ii_{j+x}(\varrho)  
\end{equation}
for all $x,y,j\in\Z_d$ and $\varrho\in\sh$.
It is straightforward to verify that the joint observable $\Co(j,k) := \Ii_j^\ast (\Bo(k))$ is a covariant phase space observable.

We demonstrate this method by choosing the $\Ao_\lambda$-compatible L\"uders instrument $\Ii^L$, defined as
\begin{equation*}
\Ii^L_j(\varrho) = \sqrt{\Ao_\lambda(j)} \varrho \sqrt{\Ao_\lambda(j)}  \, .
\end{equation*}
It is straightforward to see that $\Ii^L$ satisfies \eqref{eq:ins-cov}.
The covariant joint observable is then
\begin{equation*}
\Co(j,k) = \sqrt{\Ao_\lambda(j)} \Bo(k) \sqrt{\Ao_\lambda(j)} \, ,
\end{equation*}
and its associated state is
$$
T = d \Co(0,0) = d \sqrt{\Ao_\lambda(0)} \kb{\psi_0}{\psi_0} \sqrt{\Ao_\lambda(0)} \, .
$$
Since
\begin{eqnarray*}
\sqrt{\Ao_\lambda(j)} & = & \sqrt{\frac{1-\lam}{d}} \, \id + \left( \sqrt{\frac{(d-1) \lam + 1}{d}} - \sqrt{\frac{1-\lam}{d}} \right) \, \Ao(j) \\
& = & \frac{\beta_\lam}{\sqrt{d}} \, \id + \alpha_\lam \kb{\fii_j}{\fii_j}
\end{eqnarray*}
and
$$
\sqrt{\Ao_\lambda(0)} \psi_0 = \frac{1}{\sqrt{d}} \left(\beta_\lam \psi_0 + \alpha_\lam \fii_0\right) = \frac{1}{\sqrt{d}} \chi_\lam \, ,
$$
we see that $T=\kb{\chi_\lam}{\chi_\lam}$, hence by Proposition \ref{Prop:gammamax} the marginal $\Bo_\gamma$ is such that $\gamma=\gamma_{\max}(\lambda)$.

In conclusion, this type of sequential measurement of $\Ao_\lambda$ and $\Bo$ is effectively a joint measurement of $\Ao_\lambda$ and $\Bo_\gamma$ with minimal unsharpnesses.

\section{Discussion}\label{sec:discussion}

In our investigation we have concentrated on canonically conjugated pairs of observables, i.e., the orthonormal bases $\{\varphi_j\}_{j\in\Z_d}$ and $\{\psi_k\}_{k\in\Z_d}$ have been assumed to be Fourier connected with respect to the Fourier transform of the cyclic group $\Z_d$; see \eqref{eq:defF}. Equivalently, we have assumed that the two bases satisfy $\ip{\fii_j}{\psi_k} = (1/\sqrt{d}) \, \omega^{jk}$ for all $j,k\in\Z_d$.
As a consequence the observables $\Ao$ and $\Bo$, both defined on $\Omega_\Ao = \Omega_\Bo = \Z_d$ as $\Ao(j)=\kb{\varphi_j}{\varphi_j}$ and $\Bo(k)=\kb{\psi_k}{\psi_k}$, satisfy the covariance and invariance conditions \eqref{eq:cov-A} -- \eqref{eq:cov-B}, which turn out to be very useful in our calculations.

Our approach covers more cases than it may seem at the first sight.
Namely, we recall that two orthonormal bases $\{\varphi_j\}_j$ and $\{\varphi'_j\}_j$ define the same observable iff there are complex numbers $\alpha_j$ with $|\alpha_j|=1$ such that $\varphi'_j=\alpha_j \varphi_j$.
To illustrate an application of this many-to-one correspondence, suppose that the dimension $d$ is an odd prime number, say $d=p$ (the generalization to the case $d=p^r$, with $r$ positive integer, is straightforward). 
In this case it is easy to give a full set of $p+1$ MUBs \cite{WoFi89}; fix an orthonormal basis $\{\fii_j\}_{j\in\Z_p}$ and define $p$ orthonormal bases $\{\psi^a_k\}_{k\in\Z_p}$, each one labeled by $a\in\Z_p$, by
$$
\psi^a_k = \frac{1}{\sqrt{p}} \sum_{x\in\Z_p} \omega^{ax^2 + kx} \fii_x \, .
$$
The fact that these are MUBs follows from the Gauss summation formula
\begin{equation}\label{eq:Gauss}
\frac{1}{\sqrt{p}} \sum_{x\in\Z_p} \omega^{ax^2} = \left( \frac{a}{p} \right) \times \left\{ \begin{array}{ccc} 1 & \mbox{if} & p \in 4 \nat + 1 \\ i & \mbox{if} & p \in 4 \nat - 1 \end{array} \right. \, ,
\end{equation}
where $\left( \frac{a}{p} \right)$ is the Legendre symbol (see e.g. \cite{BeEv81}).

It is immediate to see that the orthonormal basis $\{\psi^a_k\}_{k\in\Z_p}$ is Fourier connected to the orthonormal basis $\{\fii'_j\}_{j\in\Z_p}$ given by
$$
\fii'_j = \omega^{aj^2} \fii_j \quad \forall j \, ,
$$
i.e., $\ip{\fii'_j}{\psi^a_k} = (1/\sqrt{p}) \, \omega^{jk}$.
Moreover, for $a,b\in\Z_p \setminus \{0\}$, with $a\neq b$, define the rescaled orthonormal bases $\{\psi^{a\prime}_j\}_{j\in\Z_p}$ and $\{\psi^{b\prime}_k\}_{k\in\Z_p}$, given by
\begin{align*}
\psi^{a\prime}_j & = \omega^{-4^{-1} j^2 (b-a)^{-1}} \psi^a_j \\
\psi^{b\prime}_k & = \omega^{k^2(b-a)} \left( \frac{b-a}{p} \right) \psi^b_{2k(b-a)} \times \left\{ \begin{array}{ccc} 1 & \mbox{if} & p \in 4 \nat + 1 \\ - i & \mbox{if} & p \in 4 \nat - 1 \end{array} \right.
\end{align*}
(Here $\omega^{x^{-1}}$ means `$\omega$ to the inverse of $x$ in the field $\Z_p$' and should not be confused with $e^{\frac{2\pi i}{px}}$). Then, an easy computation using the Gauss formula \eqref{eq:Gauss} yelds
$$
\ip{\psi^{b\prime}_h}{\psi^{a\prime}_k} = \frac{1}{\sqrt{p}} \omega^{-hk} \, ,
$$
which shows that also $\{\psi^{a\prime}_j\}_{j\in\Z_p}$ and $\{\psi^{b\prime}_k\}_{k\in\Z_p}$ are Fourier connected.

More generally, one can start from a complementary pair of observables, which means that $\{\varphi_j\}_j$ and $\{\psi_k\}_k$ are mutually unbiased but not necessarily Fourier connected.
Obviously, we can still ask similar questions on joint measurements.
Especially, it would be interesting to know whether Proposition \ref{prop:symmetric} is still valid under this more general setting.
In other words, the question is whether all complementary pairs are essentially similar with respect to joint measurability

Even if we leave this question open in the general case, we can see that our approach generalizes to a larger domain than we have explicitly used it for.
Indeed, all our results are still valid (and with only very slight modifications in some of the proofs) if we consider Fourier transform with respect to a generic abelian group $G$ with order $d$, i.e., $G=\Z_{d_1}\times\ldots\times\Z_{d_k}$ for $d_1+\cdots+d_n=d$ and $d_i = p_i^{r_i}$, with $p_i$ prime and $r_i$ integer for all $i=1,\ldots ,n$. 
In this case, $\hh = \hh_1 \otimes \ldots \otimes \hh_n$ with $\dim \hh_i = d_i$, a basis $\{\fii^i_j\}_{j\in\Z_{d_i}}$ is chosen in each factor Hilbert space $\hh_i$, and the $G$-Fourier transform of $\hh$ is just the tensor product $\ff = \ff_1 \otimes \ldots \otimes \ff_n$, where each $\ff_i$ is the $\Z_{d_i}$-Fourier transform in $\hh_i$ with respect to the basis $\{\fii^i_j\}_{j\in\Z_{d_i}}$, as defined in \eqref{eq:defF}.
 The mutually unbiased bases $\{\fii_j\}_{j\in\Z_d}$ and $\{\psi_k\}_{k\in\Z_d}$ are replaced by the bases $\{\fii_{j_1 \, , \ldots,\, j_n}\}_{j_1\in\Z_{d_1} \, , \ldots , \, j_n\in\Z_{d_n}}$ and $\{\psi_{k_1 \, , \ldots,\, k_n}\}_{k_1\in\Z_{d_1} \, , \ldots ,\, k_n\in\Z_{d_n}}$ of $\hh$, given by
\begin{align*}
\fii_{j_1 \, , \ldots ,\,j_n} & = \fii^1_{j_1} \otimes \ldots \otimes \fii^n_{j_n} \\
\psi_{k_1 \, , \ldots ,\,k_n} & = \ff^\ast (\fii_{j_1 \, , \ldots j_n}) = \psi^1_{k_1} \otimes \ldots \otimes \psi^n_{k_n} \, .
\end{align*}
Their associated complementary observables $\Ao$ and $\Bo$ are now both defined on $G$, and given by $\Ao (j_1 \, , \ldots , \, j_n) = \Ao(j_1) \otimes \ldots \otimes \Ao(j_n)$ and $\Bo (k_1 \, , \ldots , k_n) = \Bo(k_1) \otimes \ldots \otimes \Bo(k_n)$. They still satisfy the analogues of the covariance and invariance conditions \eqref{eq:cov-A} -- \eqref{eq:cov-B}, if the representations $U$ and $V$ are replaced by suitable tensor products.

To demonstrate that we can now handle larger class of complementary observables, let $\hi=\C^4$, choose an orthonormal basis $\{ \varphi_j \}_{j\in\{0,\ldots, 3\}}$ of $\C^4$ and set
\begin{align*}
&\psi_0 = \half ( \varphi_0 + \varphi_1  + \varphi_2 + \varphi_3) \, , & \psi_1 = \half ( \varphi_0 - \varphi_1  + \varphi_2 - \varphi_3) \, ,\\
&\psi_2= \half ( \varphi_0 + \varphi_1  - \varphi_2 - \varphi_3)\, , 
&\psi_3 = \half ( \varphi_0 - \varphi_1  - \varphi_2 + \varphi_3) \, .
\end{align*}
Then the observables $\Ao(j)=\kb{\varphi_j}{\varphi_j}$ and $\Bo(k)=\kb{\psi_k}{\psi_k}$ are complementary.
They are equally defined by any two orthonormal bases $\{ \alpha_j\varphi_j \}_{j\in\{0,\ldots, 3\}}$ and $\{\beta_k\psi_k\}_{k\in\{0,\ldots, 3\}}$, where $\alpha_j,\beta_k$ are complex numbers with unit modulus.
If some pair of these orthonormal bases were connected by the $\Z_4$-Fourier transform, then the matrix of their scalar products $[\overline{\alpha_j} \beta_k \ip{\fii_j}{\psi_k}]$ should be equal to the $\Z_4$-Fourier matrix
\begin{align*}
\frac{1}{2} \begin{pmatrix} 1&1&1&1\\ 1&i&-1&-i \\  1&-1&1&-1\\ 1&-i&-1&i \\
\end{pmatrix}
\end{align*}
or to a matrix obtained from the above by some permutations of its rows and columns. It is straightforward to verify that the deriving set of equations for $\alpha_j$ and $\beta_k$ has no solution.
However, the matrix of scalar products $[\ip{\fii_j}{\psi_k}]$ is just the Fourier matrix of $G=\Z_2 \times \Z_2$, i.e.,
\begin{equation}\label{eq:Fmatrix}
\frac{1}{2} \left( \begin{array}{cccc} 1 & 1 & 1 & 1 \\ 1 & -1 & 1 & -1 \\ 1 & 1 & -1 & -1 \\ 1 & -1 & -1 & 1 \end{array} \right) \, .
\end{equation}
In other words, the two orthonormal bases are connected by the Fourier transform of $\Z_2 \times \Z_2$.

\section*{Appendix: criterion for informational completeness}

\begin{theorem}\label{th:infocom}
Let $\Co_T$ be a covariant phase space observable.
Then $\Co_T$ is informationally complete if and only if
\begin{equation}\label{eq:icbis}
\tr{TU_x V_y} \neq 0 \quad \forall x,y\in\Z_d \, .
\end{equation}
\end{theorem}

Our proof of the above theorem relies on the following well known reconstruction formula for the Weyl-Heisenberg group, which is just a special case of orthogonality relations for irreducible representations of compact groups. 

\begin{proposition}\label{Prop:recons}
The following reconstruction formula holds for every $A\in\lh$:
\begin{equation}\label{eq:recons}
\frac{1}{d} \sum_{x,y\in\Z_d} \tr{AV^*_y U^*_x} U_x V_y = A \, .
\end{equation}
\end{proposition}

\begin{proof}
For all $h,k\in\Z_d$, we have $\ip{\fii_k}{U_x V_y \fii_h} = \omega^{hy} \, \delta_{h+x,k}$. 
Thus, for all $h,k,m,n\in\Z_d$, we obtain
\begin{eqnarray*}
&& \ip{\fii_k}{\left[\sum_{x,y\in\Z_d} \tr{\kb{\fii_m}{\fii_n} V_y^\ast U_x^\ast} U_x V_y \right] \fii_h} = \\
&& \qquad \qquad \qquad \qquad = \sum_{x,y\in\Z_d} \ip{U_x V_y \fii_n}{\fii_m} \ip{\fii_k}{U_x V_y \fii_h} \\
&& \qquad \qquad \qquad \qquad = \sum_{x,y\in\Z_d} \omega^{(h-n)y} \, \delta_{h+x,k} \, \delta_{n+x,m} \\
&& \qquad \qquad \qquad \qquad = \sum_{y\in\Z_d} \omega^{(h-n)y} \, \delta_{h-n,k-m} \\
&& \qquad \qquad \qquad \qquad = d \, \delta_{h,n} \, \delta_{h-n,k-m} \\
&& \qquad \qquad \qquad \qquad = d \, \delta_{h,n} \, \delta_{k,m} \\
&& \qquad \qquad \qquad \qquad = d \ip{\fii_k}{\fii_m} \ip{\fii_n}{\fii_h} \, ,
\end{eqnarray*}
which proves \eqref{eq:recons} for $A=\kb{\fii_m}{\fii_n}$. 
Since every $A\in\lh$ is a linear combination of this type of operators, the claim follows.
\end{proof}

\begin{proof}[Proof of Theorem \ref{th:infocom}]
Let $\ell(\Z_d^2)$ be the linear space of complex functions on $\Z_d^2\equiv\Z_d\times\Z_d$. 
We recall that, by Proposition 5.1 in \cite{BuCaLa95}, $\Co_T$ is informationally complete if and only if the linear map
$$
V_T : \lh \frecc \ell (\Z^2_d) \, , \qquad [V_T (A)] (x,y) = \tr{\Co_T (x,y) A}
$$
is injective. 
Since the dimensions of  $\lh$ and $\ell (\Z^2_d)$ are both $d^2$, we conclude that $\Co_T$ is informationally complete if and only if $V_T$ is an isomorphism.

We define the following three linear maps
\begin{align*}
\Phi & : \lh\frecc \hh\otimes\hh \, , \qquad \Phi(A) = \frac{1}{d} \sum_{x,y\in\Z_d} \tr{AV^*_y U^*_x} \, \fii_x\otimes\fii_y \\
M_T & : \hh\otimes\hh \frecc \hh\otimes\hh \, , \qquad M_T (\fii_x\otimes\fii_y) = \tr{TU_x V_y} \, \fii_x\otimes\fii_y \\
R & : \hh\otimes\hh \frecc \ell (\Z^2_d) \, , \qquad R \phi (x,y) = \ip{\fii_y\otimes\fii_x}{\phi} \, .
\end{align*}
The map $R$ is clearly a linear isomorphism, $\Phi$ is a linear isomorphism by Proposition \ref{Prop:recons}, and $M_T$ is a linear isomorphism if and only if \eqref{eq:icbis} holds. 
We now evaluate the composition map $R(\ff\otimes\ff^* )M_T \Phi$. 
For all $A\in\lh$, we obtain
\begin{eqnarray*}
&& [R(\ff\otimes\ff^* )M_T \Phi (A)](h,k) = \\
&& \qquad = \frac{1}{d} \sum_{x,y\in\Z_d} \tr{TU_x V_y} \tr{AV^*_y U^*_x} \ip{\fii_k\otimes\fii_h}{(\ff\otimes\ff^*)(\fii_x\otimes\fii_y)} \\
&& \qquad = \frac{1}{d^2} \sum_{x,y\in\Z_d} \tr{TU_x V_y} \tr{AV^*_y U^*_x} \omega^{yh-xk} \\
&& \qquad = \frac{1}{d^2} \sum_{x',y'\in\Z_d} \omega^{y'h-x'k} \tr{TU^*_h U_{x'} V^*_k V_{y'}} \tr{AV_k V^*_{y'} U_h U^*_{x'}} \\
&& \qquad = \frac{1}{d^2} \sum_{x',y'\in\Z_d} \tr{TU^*_h V^*_k U_{x'} V_{y'}} \tr{AV_k U_h V^*_{y'} U^*_{x'}} \\
&& \qquad = \frac{1}{d} \tr{TU^*_h V^*_k AV_k U_h} \\
&& \qquad = [V_T (A)] (h,k)
\end{eqnarray*}
(in the third equality we set $x=x'-h$, $y=y'-k$, in the fourth we used the commutation relation for $U$ and $V$, and in the fifth we applied the reconstruction formula \eqref{eq:recons}). As the map $R(\ff\otimes\ff^* )M_T \Phi$ is an isomorphism if and only if \eqref{eq:icbis} holds, the same is true for the map $V_T$, and the theorem follows.
\end{proof}

\section*{Acknowledgements}

T.H. is grateful to Cosmo Lupo and Mario Ziman for illuminating discussions on qubit tomography.
T.H. acknowledges financial support from the Academy of Finland (grant no. 138135) and the Magnus Ehrnrooth foundation.

\end{document}